\documentclass[11pt,fleqn,reqno]{article}

\usepackage{amsmath}
\usepackage{amssymb}
\usepackage{amsfonts}
\usepackage{amsthm}
\usepackage{mathtools}
\mathtoolsset{
	above-intertext-sep = -1ex
}
\usepackage{array}
\usepackage{makecell}
\usepackage{multirow}
\usepackage[utf8]{inputenc}
\usepackage[T1]{fontenc}

\usepackage[sans]{dsfont}

\usepackage[left=2.5cm, right=2.5cm, top=2.5cm, bottom=2.5cm]{geometry}
\usepackage[%
cal=euler,
bb=fourier,
scr=zapfc
]
{mathalfa}

\usepackage[font=small,labelfont=bf]{caption}
\usepackage{subcaption}
\usepackage{graphicx}
\usepackage{framed}
\usepackage{acronym}
\usepackage{latexsym}
\usepackage{paralist}
\usepackage{wasysym}
\usepackage{xspace}
\usepackage{framed}
\usepackage{palatino,pxfonts}
\usepackage{authblk}
\usepackage[numbers,sort&compress]{natbib}

\usepackage[dvipsnames,svgnames]{xcolor}
\colorlet{MyBlue}{DodgerBlue!75!Black}
\colorlet{MyGreen}{DarkGreen!95!Black}

\usepackage{hyperref}
\hypersetup{
	colorlinks=true,
	linktocpage=true,
	pdfstartview=FitH,
	breaklinks=true,
	pdfpagemode=UseNone,
	pageanchor=true,
	pdfpagemode=UseOutlines,
	plainpages=false,
	bookmarksnumbered,
	bookmarksopen=false,
	bookmarksopenlevel=1,
	hypertexnames=true,
	pdfhighlight=/O,
	urlcolor=MyBlue!60!black,linkcolor=MyBlue!70!black,citecolor=DarkGreen!70!black, % <--- for screen
	pdftitle={},
	pdfauthor={},
	pdfsubject={},
	pdfkeywords={},
	pdfcreator={pdfLaTeX},
	pdfproducer={LaTeX with hyperref}
}
\numberwithin{equation}{section}  %numberwithin goes before cleverefs when using hyperref
\usepackage[sort&compress,capitalize,nameinlink]{cleveref}
\crefname{example}{Ex.}{Exs.}

\crefrangeformat{equation}{\upshape(#3#1#4)\textendash(#5#2#6)}

\newcommand{\dd}{\:d}

\newcommand{\dif}{\dd}

\DeclareMathOperator*{\argmin}{argmin}

\DeclareMathOperator{\dom}{dom}

\renewcommand{\emptyset}{\varnothing}

\newcommand{\wlim}{\rightharpoonup}
\newcommand{\scrE}{\mathcal{E}}

\newcommand{\scrH}{\mathcal{H}}

\newcommand{\scrP}{\mathcal{P}}

\newcommand{\scrV}{\mathcal{V}}
\newcommand{\scrW}{\mathcal{W}}
\newcommand{\scrX}{\mathcal{X}}

\newcommand{\setH}{\mathsf{H}}
\newcommand{\Leb}{{\mathsf{Leb}}}

\newcommand{\R}{\mathbb{R}}

\newcommand{\N}{\mathbb{N}}

\DeclareMathOperator{\VI}{VI}
\DeclareMathOperator{\Sol}{SOL}

\DeclareMathOperator*{\essinf}{ess\,inf}

\usepackage{algorithmic}
\usepackage[ruled,vlined]{algorithm2e}
\theoremstyle{plain}
\newtheorem{theorem}{Theorem}

\newtheorem*{corollary*}{Corollary}
\newtheorem{lemma}[theorem]{Lemma}

\theoremstyle{definition}
\newtheorem{definition}[theorem]{Definition}
\newtheorem*{definition*}{Definition}
\newtheorem{assumption}{Assumption}
\theoremstyle{remark}
\newtheorem{remark}{Remark}
\newtheorem*{remark*}{Remark}
\newtheorem*{notation*}{Notational remark}

\numberwithin{theorem}{section}
\numberwithin{remark}{section}
\numberwithin{example}{section}

\DeclarePairedDelimiter{\abs}{\lvert}{\rvert}

\DeclarePairedDelimiter{\inner}{\langle}{\rangle}
\DeclarePairedDelimiter{\norm}{\lVert}{\rVert}

\newacro{HBA}[HBA]{Hessian-barrier algorithm}
\newacro{AHBA}[AHBA]{Adaptive-Hessian-barrier algorithm}
\newacroplural{NE}[NE]{Nash equilibria}
\newacro{VI}{variational inequality}
\newacroplural{VI}[VIs]{variational inequalities}
\newacro{DNL}{Dynamic Network Loading}
\newacro{DTA}{Dynamic Traffic Assignment}
\newacro{DUE}{Dynamic User Equilibrium}
\newacro{DSO}{Dynamic System Optimal}

\begin{document}

\title{Computing Dynamic User Equilibrium on Large-Scale Networks Without Knowing Global Parameters%: Strongly convergent adaptive algorithms under weak monotonicity assumptions} 
}
\date{\today}

\author[1]{\small Duong Viet Thong}
\author[2,3]{\small Aviv Gibali}
\author[4]{\small Mathias Staudigl }
\author[5]{\small Phan Tu Vuong}

\affil[1]{\footnotesize Division of Applied Mathematics, Thu Dau Mot University, Binh
	Duong Province, Vietnam\\
	(\href{mailto:duongvietthong@tdmu.edu.vn}{duongvietthong@tdmu.edu.vn})}
\affil[2]{\footnotesize Department of Mathematics, ORT Braude College, P.O. Box 78, Karmiel 2161002, Israel}
\affil[3]{\footnotesize The Center for Mathematics and Scientific Computation, U. Haifa, Mt. Carmel, Haifa, Israel\\
	(\href{mailto:avivg@braude.ac.il}{avivg@braude.ac.il})}
\affil[4]{\footnotesize Department of Data Science and Knowledge Engineering, Maastricht University, P.O. Box 616, NL\textendash 6200 MD Maastricht, The Netherlands\\
	(\href{mailto:m.staudigl@maastrichtuniversity.nl}{m.staudigl@maastrichtuniversity.nl})}
\affil[5]{\footnotesize Mathematical Sciences, University of Southampton, Highfield Southampton SO17 1BJ, United Kingdom,
	(\href{mailto:T.V.Phan@soton.ac.uk}{T.V.Phan@soton.ac.uk})}

\maketitle

\begin{abstract}
Dynamic user equilibrium (DUE) is a Nash-like solution concept describing an equilibrium in dynamic traffic systems over a fixed planning period. DUE is a challenging class of equilibrium problems, connecting network loading models and notions of system equilibrium in one concise mathematical framework. Recently, Friesz and Han introduced an integrated framework for DUE computation on large-scale networks, featuring a basic fixed-point algorithm for the effective computation of DUE. In the same work, they present an open-source MATLAB toolbox which allows researchers to test and validate new numerical solvers. This paper builds on this seminal contribution, and extends it in several important ways. At a conceptual level, we provide new strongly convergent algorithms designed to compute a DUE directly in the infinite-dimensional space of path flows. An important feature of our algorithms is that they give provable convergence guarantees without knowledge of global parameters. In fact, the algorithms we propose are adaptive, in the sense that they do not need a priori knowledge of global parameters of the delay operator, and which are provable convergent even for delay operators which are non-monotone. We implement our numerical schemes on standard test instances, and compare them with the numerical solution strategy employed by Friesz and Han.

\end{abstract}

\section{Introduction}
\label{intro}

This paper is concerned with a class of models known as \ac{DUE}. \ac{DUE} problems have been studied within the broader context of \ac{DTA}, which is concerned with modeling time-varying traffic flows consistent with established traffic flow theory. \ac{DTA} models are greatly influenced by Wardrop's equilibrium principle \citep{War52}, which is seen as a Nash-like equilibrium condition in an aggregative game: 
\begin{itemize}
\item[(a)] Wardrop's first principle, also known as the user optimality principle, states that road segments used in an equilibrium should display the same travel costs (i.e. delay);
\item[(b)] Wardrop's second principle, known as the system's optimality principle, assumes that drivers behave cooperatively, in making travel decisions so that the over system costs (aggregate delays) are minimized. 
\end{itemize}
Logically, the behavioral maxims (a) and (b) are disconnected, and a substantive literature in transportation research is concerned with the design of computational architectures aligning these potentially conflicting principles. Since the seminal work of \cite{MerNem78a,MerNem78b}, dynamic extensions of Wardrop's principles have paved the way to the introduction of notions like \ac{DUE} and \ac{DSO} models. For comprehensive reviews of \ac{DTA} models, we refer to \cite{PeeZil01,Jei07,Wan18}. 

In the last two decades there have been many efforts to develop a theoretically and sound formulation of \ac{DUE}, acceptable to modelers and practitioners alike. Analytical \ac{DUE} models tend to be of two varieties: (1) Route Choice (RC) DUE \citep{FriLuqTobWie89,MerNem78a,MerNem78b,ZhuMar00}, and (2) Simultaneous Route and Departure Choice (SRDC) DUE \citep{Frie93,BerFrieSuoTob01,FriKimKwoRig11,RanHalBoy96}. Both types of DUE rest on two pillars: 
\begin{enumerate}
\item A mathematical notion of equilibrium;
\item A model of network performance, based on some physical laws describing traffic flows. 
\end{enumerate}
The second pillar is known in the literature as \ac{DNL}. Equilibrium is usually expressed in terms of Wardrop's first principle. Mathematical approaches to describe equilibrium contain variational inequalities (VI) \citep{Frie93,ZhuMar00}, nonlinear complementarity problems \citep{PanHanRamUkk12,HanUkkDoa11}, differential variational inequalities \citep{PanSte08,FrieMoo06} and fixed point problems \citep{FriKimKwoRig11}. In this paper we choose the VI formulation of \ac{DUE}, and our aim is to advance computational techniques for the practical solution of \ac{DUE}. Our research builds on, and extends, recent advances in computational approaches to \ac{DUE} reported in \cite{HanEveFri19}. As is well known computing user equilibrium is a challenging task; Its main complication arises since it constitutes an interconnected computational procedure, coupling equilibrium computation with \ac{DNL}. The \ac{DNL}, which could be understood as the first layer of the problem, aims at describing the spatial and temporal evolution of traffic flows on a network that is consistent with established route and departure choices of travelers. This is done by formulating appropriate dynamics to flow propagation, flow conservation, link delay, and path delay on a network level. In general, \ac{DNL} models have the following components:
\begin{enumerate}
\item Some form of link and/or path dynamics;
\item An computationally-friendly relationship between flow/speed/density and link traversal time;
\item Flow propagation constraints;
\item A model of junction dynamics (Riemann Solvers) and delays;
\item A model of path traversal times, and 
\item Appropriate initial conditions.
\end{enumerate}

\ac{DNL} generates the \emph{path delay operator}, which is the key input when computing an equilibrium given the delays on user routes (travel costs). This is the second layer of the problem, and of main interest in this paper. At this layer one has to use some equilibrium solver, whose performance depends significantly on the information we have about the structural properties of the delay operator. However, since the delay operator is itself the result of a computational procedure, it is not available in closed form, and thus one is confronted essentially with a black-box upon which we can assume whatever we find useful, but the empirical validation of these assumptions is very hard. It is thus of utmost importance to have at our disposal efficiently implementable algorithms which are:
\begin{itemize}
\item[(i)] Adaptive to arrival of new information about unknown global parameters;
\item[(ii)] Provably convergent under mild monotonicity assumptions. 
\end{itemize}
We argue that, up to now, none of the perceived DUE solvers meet both of these criteria. To support this claim, we present Table \ref{tab:Algorithms}, where the current state-of-the-art in \ac{DUE} computation is summarized.\footnote{In this table we focus on algorithms acting directly on the infinite-dimensional Hilbert space formulation of \ac{DUE}. A much larger literature on this topic exists which is concerned with finite-dimensional approximations. In the parlance of numerical mathematics, the latter would correspond to a \emph{first discretize, then optimize} strategy. As the two approaches are quite different, it would not provide fair comparisons.}

%%%%%%%Table Algorithms%%%
 \begin{table}[t!]
\centering
\begin{tabular}{|c|c|c|c|c|}
 \hline
 \thead{Algorithm} & \thead{DUE Model}  & \thead{Assumptions} & \thead{Convergence} & \thead{References}\\
 \hline
 
Projected Gradient  & SRDT  & \makecell{Lipschitz cont. \\  strongly monotone} & strong  &  \cite{FriKimKwoRig11} \\
\hline

descent algorithm & SRDT & Co-coercive & weak  & \cite{SzeLo04} \\
\hline
Route-swapping &  RC DUE & monotone   &  weak & \cite{SzeLo06}\\
\hline
Route-swapping & SRDT DUE &\makecell{Continuous\\ monotone} & weak & \cite{HuaLam02}\\
\hline
Route-swapping & SRDT DUE & \makecell{Continuous \\ monotone} & weak  &  \cite{TiaHuaGao12}\\
\hline
Extragradient &  RC DUE & \makecell{ Lipschitz cont.\\ pseudo monotone} & weak & \cite{Long2013} \\
\hline
 Self-adaptive & SRDT DUE & D-property & weak  & \cite{HanSzeFrie15}\\
\hline
Proximal point  & SRDT DUE &  Dual solvable & weak  & \cite{FriHanLiuSze15}\\
\hline
FBF & SRDT DUE & \makecell{ Lipschitz cont.\\ pseudo monotone} & strong & \cite{DuvMeiStaVu20}\\
\hline 
Inertial-FBF & SRDT DUE & \makecell{ Lipschitz cont.\\ pseudo monotone} & strong & This paper\\
\hline
\end{tabular}
\caption{Computational algorithms for DUE (adapted from \cite{HanEveFri19}). The algorithms are arranged in an increasing order of generality of the monotonicity.}
\label{tab:Algorithms}
\end{table}
%%%%%%%%%%%%%%%%%%%%%%%%%
We infer from Table \ref{tab:Algorithms} that known algorithmic strategies for solving the \ac{DUE} problem require knowledge about the global Lipschitz constant and some sort of monotonicity of the path delay operator. Since the delay operator is not given to us in closed form, both assumptions are practically not verifiable. Algorithmic strategies which are provably convergent without explicit knowledge of these global properties, are thus to be seen as a very valuable contribution. 
 
\subsection{Our Contributions}
This paper makes a significant step-ahead relative to the perceived computational literature on \ac{DUE}, by describing two numerical algorithms acting directly in infinite-dimensional Hilbert spaces. Our algorithms share the following features: 
\begin{enumerate}
\item[(i)] \emph{Strong convergence} to a single user equilibrium;
\item[(ii)] \emph{Adaptive step-size} choices without the need to know global Lipschitz parameters of the delay operator;
\item[(iii)] Provably convergent under a plain \emph{pseudo-monotonicity} assumption on the path delay operator. 
\item[(iv)] Include \emph{inertial and relaxation effects} to potentially speed up the convergence.
\end{enumerate}
While items (ii) and (iii) don't need much motivation, our emphasis on strongly convergent methods seems to be somewhat pedantic at first sight, so it deserves some words of explanation. 

In infinite-dimensional settings strongly convergent iterative schemes are much more desirable than weakly convergent ones since strong convergence translates the physically tangible property that the energy $\norm{h_{n}-h^{\ast}}^{2}$ of the error between the iterate $h_{n}$ and a solution $h^{\ast}$ eventually becomes arbitrarily small. Of course, any numerical solution technique designed for solving a problem in infinite dimensions must be applied to a finite-dimensional approximation of the problem. Exactly in such situations strongly convergent methods are extremely powerful, because they guarantee stability with respect to numerical discretization. In fact, \cite{Guel91} demonstrated that strongly convergent schemes might even exhibit faster convergence rates as compared to their weakly convergent counterparts. It seems therefore fair to say that strong convergence is an extremely desirable property of solution schemes, with clearly observable physical consequences on the performance and stability of algorithms. As a matter of fact, \cite{FriKimKwoRig11} employs a projected gradient iteration of Halpern type \citep{halpern1967,BauCom16}, which forces trajectories to converge strongly to some \ac{DUE}. 

Adaptivity in the step-size policy frees us from any unavailable information about the global Lipschitz constant of the delay operator. It allows us to tune the step size on-the-fly and guarantees convergence for general pseudo-monotone operators with good performance properties.  

Operator splitting methods with inertia and relaxation have received quite some attention in recent years, see e.g. \cite{LorPoc15,IutHen19,AttCab19}. These schemes are motivated by Nesterov's accelerated method \citep{NesConvex}, and therefore the main motivation for inertial methods is to speed up the convergence rate. To the best of our knowledge this is the first time that inertial and relaxation effects are investigated in the context of \ac{DUE} computation and under weak pseudo-monotonicity assumptions.\\ 
\begin{remark}
In previous work \citep{DuvMeiStaVu20} investigated the \ac{DUE} with a strongly convergent FBF variant. This paper replaces and significantly extends our previous work by the explicit consideration of inertial effects. 
\end{remark}
\subsection{Organization of the paper}
Sections \ref{sec:DUE} and \ref{sec:DNL} describe user equilibrium and the DNL procedure we use in our numerical experiments. In setting up these two layers we follow closely \cite{HanEveFri19}. Section \ref{sec:algo} describes the algorithms we construct and investigate in this paper. Building on the MATLAB toolbox publicly available at \url{https://github.com/DrKeHan/DTA} and documented in \cite{HanEveFri19}. We report the outcomes of our experiments in Section \ref{sec:numerics}. Technical facts and proofs are organized in Sections \ref{sec:prelims} and \ref{sec:ProofMain}.

%----------------------------------------------------------------------
%%% PRELIMS
%----------------------------------------------------------------------
\section{Dynamic User Equilibrium}
\label{sec:DUE}
We introduce a few notations and terminologies for the ease of presentation below. 

\begin{itemize}
\item $\scrP$: set of paths in the network.
\item $\scrW$: set of origin-destination (O-D) pairs in the network.
\item $Q_{w}$: fixed O-D demand between $w\in\scrW$.
\item $\scrP_{w}$: subset of paths that connect O-D pair $w$.
\item $t$: continuous time parameter in the fixed time horizon $[t_{0},t_{1}]$.
\item $h_{p}(t)$: departure rate along path $p$ at time $t$.
\item $h(t)$: complete profile of departure rates $h(t)=\{h_{p}(t);p\in\scrP\}$.
\item $A_{p}(t,h)$: effective travel cost along path $p$ with departure time $t$ under the path profile $h$.
\item $\nu_{w}(h)$: minimum travel cost between O-D pair $w\in\scrW$ for all paths and all departure times.
\end{itemize}

\subsection{Formulation of DUE as a Variational inequality}
Let $[t_{0},t_{1}]$ be a fixed planning horizon. We are given a connected directed graph $G=(\scrV,\scrE)$ with finite set of vertices $\scrV$, representing traffic intersections (junctions) and arc set $\scrE$, representing road segments. A path $p$ in the graph $G$ is identified with a non-repeating finite sequence of links it traverses, i.e. $p=\{I_{1},I_{2},\ldots,I_{m(p)}\},$ where $m(p)$ is the number of links in this path. We denote the set of all paths by $\scrP$, and set $\setH:= \R^{\abs{\scrP}}$. We are interested in paths which connect a set of distinguished vertices acting as the \emph{origin-destination} (O-D) pairs in our graph. We are given $N$ distinct O-D pairs denoted as $w_{1},\ldots,w_{N}$, where each $w_{i}=(o_{i},d_{i})\in\scrV$. Call $\scrW:=\{w_{1},\ldots,w_{N}\}$ the collection of all O-D pairs, and let us denote the set of paths connecting the O-D pair $w$ by $\scrP_{w}\subseteq\scrP.$ For each O-D pair $w\in\scrW$ we are given an exogenous \emph{demand} $Q_{w}>0$; This represents the number of drivers who have to travel from the origin to the destination described by $w$. The list $Q=(Q_{w})_{w\in\scrW}$ is often called the \emph{trip table}. In DUE modeling, the single most crucial ingredient is the path delay operator, which maps a given vector of departure rates (path flows) $h$ to a vector of path travel times. We stipulate that path flows are square integrable functions over the planning horizon, so that $h_{p}\in L^{2}([t_{0},t_{1}];\R_{+})$ and $h=(h_{p};p\in\scrP)\in \scrH:= L^{2}([t_{0},t_{1}];\setH)$. To measure the delay of drivers on paths, we introduce the operator $D:\scrH\to\scrH,h\mapsto D(h)$, with the interpretation that $D_{p}(t,h)$ is the path travel time of a driver departing at time $t$ from the origin of path $p$, and following this path throughout. This operator is the result of some \ac{DNL} procedure, which is an integrated subroutine in the dynamic traffic assignment problem. See Section \ref{sec:DNL} for a description of the \ac{DNL} used in our computational experiments.\\

On top of path delays, we consider penalty terms of the form $\phi(t+D_{p}(t,h)-\tau),$ penalizing all arrival times different from the target time $\tau>0$ (i.e. the usual time of a trip on the O-D pair $w$). The function $\phi:[-\infty,\infty)\to[0,\infty]$ should be monotonically increasing with $\phi(a)>0$ for $a>0$ and $\phi(a)=0$ for $a\leq 0$. Define the \emph{effective delay operator} as 
\begin{equation}\label{eq:A}
A_{p}(t,h):=D_{p}(t,h)+\phi(t+D_{p}(t,h)-\tau).
\end{equation}
We thus obtain an operator $A:\scrH\to\scrH$, mapping each profile of path departure rates $h$ to effective delays $A(h)=\{A_{p}(t,h);t\in[t_{0},t_{1}]\}\in\scrH$.

We follow the perceived \ac{DUE} literature, and stipulate that \emph{Wardrop's first principle} holds: Users of the network aim to minimize their own travel time, given the departure rates in the system. Thus, a user equilibrium is envisaged, where the delays (interpreted as costs) of all travelers in the same O-D pair are equal, and no traveler can lower his/her costs by unilaterally switching to a different route. To put this behavioral axiom into a mathematical framework, we first formulate the meaning of "minimal costs" in the present Hilbert space setting. Recall the essential infimum of a measurable function $g:[t_{0},t_{1}]\to\R$ as 
$\essinf\{g(t): t\in[t_{0},t_{1}]\}=\sup\left\{x\in\R: \Leb(\{s\in[t_{0},t_{1}]:g(s)<x\})=0\right\},$ where $\Leb(\cdot)$ denoted the Lebesgue measure on the real line. Given a profile $h\in\scrH$, define 
\begin{align}
\nu_{p}(h)&:=\essinf\{A_{p}(t,h): t\in[t_{0},t_{1}]\}\qquad\forall p\in\scrP, \text{ and }\\
\ nu_{w}(h)&:= \min_{p\in\scrP_{w}}\nu_{p}(h)\qquad\forall w\in\scrW.
\end{align}
On top of minimal costs, we have to restrict the set of departure rates to functions satisfying a basic flow conservation property. Specifically, insisting that all trips are realized, we naturally define the set of feasible flows as
\begin{equation}
\scrX:=\left\{f\in\scrH: \sum_{p\in\scrP_{w}}\int_{t_{0}}^{t_{1}}f_{p}(t)\dif t=Q_{w}\quad\forall w\in\scrW\right\}. 
\end{equation}
The set of feasible flows $\scrX$ is sequentially closed and convex, but not sequentially compact (i.e. path departure rates are note a-priori assumed to be bounded as the above definition involves Lebesgue-integrable functions). We are now ready to give our first definition of user equilibrium. 
%--------------DUE-Version 1-------------------
\begin{definition}
\label{def:DUE-V1}
A profile of departure rates $h^{\ast}\in\scrH$ is a \ac{DUE} if 
\begin{itemize}
\item[(a)] $h^{\ast}\in\scrX$, and 
\item[(b)] $h^{\ast}_{p}(t)>0,p\in\scrP_{w}\Rightarrow A_{p}(t,h^{\ast})=\nu_{w}(h^{\ast}).$
\end{itemize}
We denote by $\Omega\subset\scrX$ the (possibly empty) set of \ac{DUE}. 
\end{definition}
%----------------------------------------------------
In \cite{Frie93} it is observed that the definition of \ac{DUE} can be formulated equivalently as a variational inequality $\VI(A,\scrX)$: A flow $h^{\ast}\in\scrX$ is a \ac{DUE} if 
\begin{equation}\label{eq:DUE}
\inner{A(h^{\ast}),h-h^{\ast}}\geq 0\qquad\forall h\in\scrX
\end{equation}
This notion of equilibrium is very useful, since it allows us to apply a large variety of algorithms to solve $\VI(A,\scrX)$, and in fact it can be seen as the basis of most of the computational approaches to \ac{DUE}. We now spell out sufficient conditions guaranteeing existence of \ac{DUE}. 
%%%%%%%%%%
\begin{assumption}\label{ass:1}
\begin{itemize}
\item The penalty function $\phi:[t_{0},t_{1}]\to\R_{+}$ is continuous and there exists $\Delta>-1$ such that 
\begin{equation}
\phi(a)-\phi(b)\geq\Delta(a-b)\text{ for all }t_{0}\leq a<b\leq t_{1}.
\end{equation}
\item The \ac{DNL} satisfies the FIFO principle and each link has finite capacity.
\item The effective delay operator is weak-to-weak continuous on bounded subsets of $\scrX$. 
\end{itemize}
\end{assumption}

\begin{theorem}
Under Assumption \ref{ass:1} the \ac{DUE} problem \eqref{eq:DUE} has a solution, i.e. $\Omega\neq\emptyset$. 
\end{theorem}
\begin{proof}
See \cite{FriHanTau13}.

\end{proof}

The construction of the delay operator requires a specification of a \ac{DNL} (i.e. traffic flow generation). We focus in this work on a macroscopic model of network loading based on fluid dynamic approximations of traffic flow on networks, known as the Lighthill-Whitham-Richards (LWR) model \citep{LW55,Ric56}. The LWR model is able to describe the physics of kinematic waves (e.g. shock waves, rarefaction waves), and allows network extension that capture the formation and propagation of vehicle queues as well as vehicle spill-back. We will formulate the LWR-based \ac{DNL} as a system of partial differential algebraic equations (PDAE), which uses vehicle density and queues as the unknown variables, and computes link dynamics, flow propagation, and path delay for any given vector of path departure rates. 

\subsection{The differential variational inequality formulation}
It has been observed in \cite{BerFrieSuoTob01} that \ac{DUE} can be equivalently formulated as a differential variational inequality \citep{PanSte08}. From an algorithmic point-of-view this relation is interesting as it allows us to use time-stepping methods to compute approximate user equilibria \cite{FriKimKwoRig11,FrieMoo06}. Independent of algorithmic considerations, we regard this identification as an important conceptual insight, and thus deserves some remarks here. The precise connection between DVI and DUE goes as follows:\\ 

Define the vector-valued function $x:[t_{0},t_{1}]\to\R^{\abs{\scrW}},t\mapsto x(t)=\{x_{w}(t);w\in\scrW\}$ as the state trajectory of a controlled dynamical system with the interpretation that $x_{w}(t)$ is the cumulative traffic up to time $t$ on paths connecting the origin-destination pair $w\in\scrW$. The definition of this state-variable requires that its dynamic evolution is described by the linear differential equation 
\begin{equation}\label{eq:ODE}
\frac{\dif}{\dif t}x_{w}(t)=\sum_{p\in\scrP_{w}}h_{p}(t)\qquad\text{a.e.  }t \in[t_{0},t_{1}].
\end{equation}
Additionally, it must satisfy the natural initial and boundary-value conditions
\begin{equation}\label{eq:IBVP}
(x_{w}(t_{0}),x_{w}(t_{1}))=(0,d_{w})\qquad\forall w\in\scrW.
\end{equation}
The differential variational inequality describing DUE reads then as follows: Find $h\in\scrH$ such that \eqref{eq:ODE}, \eqref{eq:IBVP} and the instantaneous optimality condition 
\begin{equation}
h(t)\in\Sol(\R^{\abs{\scrP}}_{+},A(t,\cdot))\quad \text{a.e. }t\in[t_{0},t_{1}]
\end{equation}
holds. Note that this defines a time-dependent complementarity system 
\[
0\leq h(t)\bot A(t,h(t))\geq 0\quad \text{a.e. }t\in[t_{0},t_{1}],
\]
which has been used in a DUE model with a simplified bottleneck structure in \cite{PanHanRamUkk12}. See \cite{FriKimKwoRig11} for a formal proof on the correctness of this interpretation.

\section{Dynamic Network Loading}
\label{sec:DNL}
The purpose of this section is to explain the dynamic network loading model used in our numerical investigation. We are considering the LWR model on networks, adopting the description in terms of a system of Differential Algebraic Equations (DAE). This formulation of the DNL procedure has the advantage over its mathematically equivalent description in terms of a system of partial differential algebraic equations that it avoids the use of partial differential operators, and thus is much more amenable to numerical discretization strategies.

\subsection{The Lighthill-Whitham-Richards link model}
\label{sec:LWR}
Network loading acts on the same oriented graph $G=(\scrV,\scrE)$ as in Section \ref{sec:DUE}, where links $I_{i}\in\scrE$ have a certain length measured by the interval $[a_{i},b_{i}]$. The within-link dynamics are captured by the scalar conservation law 
\begin{equation}\label{eq:LWR}
\partial_{t}\rho_{i}(t,x)+\partial_{x}\left[\rho_{i}(t,x)v_{i}(\rho_{i}(t,x))\right]=0\quad (t,x)\in[t_{0},t_{1}]\times[a_{i},b_{i}].
\end{equation}
The fundamental diagram $f_{i}(\rho)=\rho\cdot v_{i}(\rho)$ is assumed to be continuous, concave and vanishes at $\rho\in\{0,\rho^{jam}_{i}\}$, where $\rho^{jam}_{i}$ is the jam density on link $I_{i}$.  Moreover, there exists a unique global maximum of $f_{i}$ at the value $\rho_{i}^{c}$. We focus on the triangular fundamental diagram 
\begin{equation}
f_{i}(\rho)=\left\{\begin{array}{ll} v_{i}\rho & \text{if }\rho\in[0,\rho^{c}_{i}],\\
-w_{i}(\rho-\rho^{jam}) & \text{if }\rho\in(\rho^{c}_{i},\rho^{jam}_{i}]
\end{array}\right.
\end{equation}
where $v_{i},w_{i}>0$ denote the forward and backward kinematic wave speeds, respectively.

At junctions we need to make sure that relevant boundary conditions are satisfied to respect basic physical principles. Consider a junction with $m$ incoming and $n$ outgoing links. At each such junction, the following conservation property must hold: 
\begin{equation}\label{eq:Kirchhoff}
\sum_{i=1}^{m}f_{i}(\rho_{i}(t,b_{i}))=\sum_{j=1}^{n}f_{j}(\rho_{j}(t,a_{j}))\qquad\forall t\in[t_{0},t_{1}].
\end{equation}
This condition simply means that inflow into the junction equals outflow. However, this condition alone does not guarantee a unique flow profile at these $m+n$ links. Additional conditions, usually formulated in terms of Riemann solvers and demand/supply conditions must be imposed. We refer to \cite{BreEMS14,GarHePic16} for reviews. 

\subsection{The variational representation of link dynamics}
\label{sec:LaxHopf}
While \eqref{eq:LWR} captures within-link dynamics, the inter-link propagation of congestion requires a careful treatment of junction dynamics. The overall system of PDEs leads to a complex system of junction dynamics and conservation laws which is very hard to handle computationally. We follow a different approach here, which is more amenable to numerical computations. We briefly introduce a variational representation of the link dynamics, based on the generalized Lax-Hopf formula, originally developed in \cite{AubBaySP08,ClaBay10PartI,ClaBay10PartII}, which leads to a DNL procedure in terms of a system of differential algebraic equations (DAE). Compared to the flow-based approach described in Section \ref{sec:LWR}, the DAE based formulation has the following main advantages: (1) the primary variable is flow instead of density); (2) no partial differential operators are involved; (3) it introduces simplified boundary conditions. We only give a high-level description of this approach, detailed enough so that the reader is able to understand the mechanics of the numerical solver. A rigorous description can be found in \cite{GarHePic16}.\\

Consider the Moskowitz function $N_{i}(t,x)$ which measures the cumulative number of vehicles that have passed location $x$ along link $I_{i}$ by time $t$. The following identities hold: 
\begin{equation}
\partial_{t}N_{i}(t,x)=f_{i}(\rho_{i}(t,x)),\quad \partial_{x}N_{i}(t,x)=-\rho_{i}(t,x).
\end{equation}
It follows immediately that $N_{i}(t,x)$ satisfies the Hamilton-Jacobi equation 
\begin{equation}
\partial_{t}N_{i}(t,x)-f_{i}(-\partial_{x}N_{i}(t,x))=0\qquad x\in[a_{i},b_{i}],t\in[t_{0},t_{1}].
\end{equation}
Denote by $f^{in}_{i}(t)$ and $f^{out}_{i}(t)$ the link $I_{i}$ inflow and outflow. The cumulative link entering and exiting vehicle counts are defined as  
$$
\frac{\dif}{\dif t}N_{i}^{up}(t)=f_{i}^{in}(t),\quad \frac{\dif}{\dif t}N^{down}_{i}(t)=f_{i}^{out}(t),
$$
where "up" and "down" represent the upstream and downstream boundaries of the link, respectively. \cite{HanPicSze16} derive explicit formulae for the
link demand and supply based on a variational formulation known as the Lax-Hopf formula \cite{AubBaySP08,ClaBay10PartI,ClaBay10PartII}, as follows:
$$
D_{i}(t)=\left\{\begin{array}{cc} -
f^{in}_{i}(t-L_{i}/v_{i}) & \text{if } N^{up}_{i}(t-L_{i}/v_{i})=N^{dn}_{i}(t)\\
C_{i} & \text{if } N_{i}^{up}(t-L_{i}/v)>N^{dn}_{i}(t)
\end{array}\right.
$$
and 
$$
S_{i}(t)=\left\{\begin{array}{cc} 
f^{out}_{i}(t-L_{i}/w_{i}) & \text{if } N^{up}_{i}(t)=N^{dn}_{i}(t-L_{i}/w_{i})+\rho^{jam}_{i}L_{i}\\
C_{i} & \text{if } N^{up}_{i}(t)< N^{dn}_{i}(t-L_{i}/w_{i})+\rho^{jam}_{i}L_{i}.
\end{array}\right.
$$

where $L_{i}=b_{i}-a_{i}$ is the length of the link $I_{i}i$, $v_{i}=f'_{i}(0+)$ and $w_{i}=f'_{i}(\rho^{jam}_{i}-)$. These two relations express the link demand and supply, which are inputs of the junction model, in terms of $N^{up}$ and $N^{down}$. This means that one no longer has to compute the dynamics within the link, but focus instead on the cumulative counts at the two boundaries of the link. Note that, when discretizing the DNL in time, we immediately obtain the link transmission model \cite{Ype05}. In general, the approach just described gives rise to the link-based formulation of DNL \cite{HanPicSze16}. 

\paragraph{Junction Dynamics} 
In a path-based DNL procedure one must incorporate established routing information into the junction model. Such information is usually formulated by some behavioral assumption on drivers' preferences. In the numerical scheme we consider, such information is provided in terms of an endogenously given flow distribution matrix $W(t)=[w_{ij}(t)]$, where $w_{ij}(t)$ is the proportion of flow incoming into link $i$ and continuing by following link $j$ at a given junction. Abstractly, if $\Theta$ represents some junction model, we have the functional relationship
$$
\left(f^{out}(t),f^{in}(t)\right)=\Theta(D(t),S(t),W(t)),
$$
where $f^{out}(t)=(f_{i}^{out}(t))_{i=1,\ldots,m}$ and $f^{in}(t)=(f_{j}^{in}(t))_{j=1,\ldots,n}$, are the computed incoming and outgoing flows.

\paragraph{Dynamics at the origin nodes}
At the origin nodes, we employ a simple point-queue model, in the spirit of \citet{Vic69}. Let $o$ be a given origin node, and denote by $q_{o}(t)$ the volume of the point queue. Let link $j$ be connected to the origin. We assume that 
$$
\frac{\dif }{\dif t}q_{o}(t)=\sum_{p\in\scrP_{o}}h_{p}(t)-\min\{D_{0}(t),S_{j}(t)\},
$$
where $\scrP_{o}$ denotes the set of paths originating from $o$. The first term on the right represents the inflow into the queue, while the second term represents flow leaving the queue, modeling the demand at the origin as 
$$
D_{o}(t)=\left\{\begin{array}{cc} 
M & \text{ if }q_{o}(t)>0,\\
\sum_{p\in\scrP_{o}}h_{p}(t) & \text{ else}
\end{array}\right.
$$
taking $M$ to be a sufficiently large number, bigger than the flow capacity at link $j$.

\paragraph{Calculating path travel times}
The DNL procedure calculates the path travel times with given path departure rates. The path travel time is defined as link travel time plus possible queuing at the origin. We define the link exit time function $\lambda(t)$ implicitly as 
\begin{equation}
N^{up}(t)=N^{down}(\lambda(t)).
\end{equation}
For a path enumerated as $p=\{1,2,\ldots,K\}$, the path travel time $D_{p}(t,h)$ is calculated as 
$$
D_{p}(t,h)=\lambda_{s}\circ\lambda_{1}\circ \ldots \circ \lambda_{m}(t).
$$
 where $(f\circ g)(t)=f(g(t))$ denotes the composition of two functions. $\lambda_{o}(t)$ is the exit time function for the potential queuing at the origin $o$. 

%----------------------------------------------------------------------
%% Algorithm
%----------------------------------------------------------------------
\section{Strongly convergent fixed-point algorithms}
\label{sec:algo}

\subsection{Fixed Point formulation of DUE}
Once a \ac{DNL} procedure has been fixed, the effective delay operator $A(h)$ can be evaluated. The definition of \ac{DUE} allows us to construct a suitable fixed-point problem which is the basis for the design of iterative numerical schemes for computing \ac{DUE}. In fact, it is easy to see that $h^{\ast}\in\scrH$ is a path-departure rate profile corresponding to a \ac{DUE} if and only if the \emph{residual} 
$$
r_{\tau}(h)=h-P_{\scrX}(h-\tau A(h))
$$
is zero, i.e. $r_{\tau}(h^{\ast})=0,\tau >0$. Here, we call $P_{\scrX}(x)$ the orthogonal projection in $L^{2}$ onto the set $\scrX\subset\scrH$. A classical iterative scheme to find the roots of a nonlinear function is the Picard fixed-point iteration to localize a fixed point of the map $h\mapsto P_{\scrX}(h-\tau A(h))$. Under strong a-priori continuity and monotonicity assumptions on the effective delay operator $A$, the \emph{projected gradient} (a.k.a. \emph{forward-backward}) method, Algorithm \ref{alg:FB}, generates a sequence $\{h_n\}_{n\in\N}$ which will weakly converge to some DUE. 
%------------ALGORITHM-FB--------------------
\begin{algorithm}[t]
	\caption{FB for $\VI(A,\scrX)$.}
	\label{alg:FB}
 \begin{algorithmic}
	\STATE {\bfseries Input:} Effective delay operator $A:\scrH\to\scrH$, step size $\{\tau_{n}\}_{n\in\N}$, Initial point $h_{0}\in\scrX$, $N\geq 1$ stopping time.	
	\FOR{$n=0,1,\ldots,N$}
	\STATE obtain $h_{n}$ by running a DNL procedure; 
	\IF{Stopping condition not satisfied}
	\STATE Update 
\begin{equation}\label{eq:FB}
h_{n+1}=P_{\scrX}(h_{n}-\tau_{n} A(h_{n})).
\end{equation}
\ENDIF
\ENDFOR
\end{algorithmic}
\end{algorithm}
%------------------------------------
This iterative solver is used in the software package developed in \cite{HanEveFri19}, and has also been employed in many studies before. Weak convergence (see Definition \ref{def:convergence} in Section \ref{sec:prelims}) of the thus constructed sequence $\{h_{n}\}_{n\in\N}$ is known when the operator $A$ is inverse strongly monotone (\emph{co-coercive}) with modulus $\mu>0$
$$
\inner{A(x) - A(y),x-y} \ge \mu \|A(x)-A(y)\|^2 \quad \forall x,y \in \scrH, 
$$
 provided that the step sizes $\tau\in(0,2\mu)$. Note that co-coercivity is equivalent to Lipschitz continuity with Lipschitz constant $\frac{1}{\mu}$. Thus, for making method \eqref{eq:FB} a provably convergent algorithm, we need to know the Lipschitz constant to pin down an upper bound on the step sizes. \emph{Strong convergence} of $\{h_{n}\}_{n\in\N}$ requires even stronger uniform monotonicity assumption of the operator $A$ over the set $\scrX$ (Theorem 25.8 \cite{BauCom16}),\footnote{An operator $A:\scrH\to\scrH$ is called \emph{uniformly monotone} if there exists an increasing function $\omega:(0,\infty)\to[0,\infty)$, vanishing at zero, such that 
$$
\inner{A(h)-A(h'),h-h'}\geq\omega(\norm{h-h'})\qquad \forall h,h'\in\dom A.
$$
}
 or other modifications of the basic template \eqref{eq:FB} are needed. \cite{FriKimKwoRig11} present a strongly convergent variant of \eqref{eq:FB} using a Halpern-type modification of the basic scheme above. Both assumptions, Lipschitz continuity and uniform monotonicity, are very restrictive in the context of computing \ac{DUE}. While continuity of the effective delay operator has been established in the context of the LWR network loading procedure \cite{HanPicFri16}, monotonicity estimates are hardly available for realistic DNL procedures and not very likely to hold in practice. Therefore, strongly convergent algorithm which are provably convergent to a solution under mild monotonicity assumptions are highly desirable for modeling, optimization and simulation of traffic networks. 

\subsection{Computing \ac{DUE} under weak assumptions}
Our aim is to design and study alternative numerical schemes for computing \ac{DUE}, which require significantly less stringent a-priori assumptions on the delay operator, but still come with rigorous convergence guarantees. We summarize our working assumptions below, while Section \ref{sec:prelims} gathers precise mathematical definitions for the readers' convenience. 
%%%%%%%%%%
\begin{assumption}
\label{ass:Lipschitz}
The delay operator $A:\scrH\to\scrH$ is sequentially weakly continuous and $L$-Lipschitz continuous on $\scrX$. However, we do not need to know $L$. 
\end{assumption}
%%%%%%%%%%%%%%
To cope with the unavailable information about the Lipschitz constant, we construct \emph{adaptive} algorithms, and thus do not need information of this hardly available global parameter. Instead a simple and efficient update procedure of the step size parameters is proposed which depends on pointwise variations of the delay operator rather than global variations. This is a major advantage of the methods we propose here, both from a conceptual and practical point of view, as it allows us to decide step-sizes "online". 

The next assumption is concerned with the structural properties with impose on the delay operator.
%%%%%%%%%%%%%%%%%%%%%%
\begin{assumption}
\label{ass:monotone}
The delay operator $A$ is pseudo-monotone on $\scrH$: For all $h_{1},h_{2}\in\scrH$, we have
\begin{equation}
\inner{A(h_{1}),h_{2}-h_{1}}\geq 0\Rightarrow \inner{A(h_{2}),h_{2}-h_{1}}\geq 0
\end{equation}
\end{assumption}
%%%%%%%%%%%%%%%%%%%
Pseudo-monotonicity is a significant weakening of the (strict) monotonicity required when applying the fixed point iteration scheme \eqref{eq:FB}. Some intuition for this concept can be given by considering the simpler case when the operator is integrable. Any smooth real-valued function $f:\scrH\to\R$ induces an operator $A:\scrH\to\scrH$ via its gradient $A(h)=\nabla f(h)$ (unique thanks to the Riesz representation theorem). Note that $f$ is (strictly) convex if and only if the gradient map is a (strictly) monotone operator. If $f$ is merely \emph{quasi-convex}, the gradient operator is pseudo-monotone and vice versa. Assumptions \ref{ass:1}-\ref{ass:monotone} are the standing hypothesis for the rest of this paper. Building on them, we now describe the numerical schemes we analyze. 

Our basic algorithmic design principle follows the \emph{forward-backward-forward} (FBF) splitting scheme, originally due to \cite{Tse00}. In its original form, it ensures that path flows will weakly converge to a \ac{DUE}, provided that the delay operator is monotone and Lipschitz continuous in the $L^{2}$ norm. In the special case of variational inequalities, it has been shown that pseudo-monotonicity suffices for weak convergence \cite{BotMerStaVuo19}. Actually, one can easily see that weak convergence holds for a large class of non-monotone VIs satisfying an angle property at the solution set \cite{DanLan15}, and this is the main reason why FBF is an attractive numerical solution scheme for \ac{DUE}. FBF updates a current path departure rate profile $h$ by first applying \eqref{eq:FB}, in order to produce the extrapolated search point $y=P_{\scrX}(h-\tau A(h))$ (first forward-backward step). It then performs another forward step in path space, by calling the DNL procedure at the just constructed extrapolation point $y$, and shifts density into directions where the difference in the ``travel costs'' between the current path flow $h$ and the new search point $y$ is large. Algebraically, this leads to the correction step $h^{+}=y+\tau(A(h)-A(y))$. We would like to emphasize that this correction step does not involve an additional projection onto the feasible set $\scrX$. This reduces the computational complexity of FBF relative to its close cousin the extragradient algorithm due to \cite{Kor76}, and speeds up the computations in practice whenever projections are expensive to implement (see \cite{BotMerStaVuo19} for extensive numerical evidence supporting this claim).  

In order to force strong convergence of the sequence of path departure rates $\{h_{n}\}_{n\in\N}$, we augment the scheme by an Halpern-type relaxation procedure. The pseudo-code of the resulting DUE solver is displayed in Algorithm \ref{alg:FBF}.

%------------ALGORITHM-FBF--------------------
\begin{algorithm}[t]
	\caption{FBF for $\VI(A,\scrX)$.}
	\label{alg:FBF}
	 \begin{algorithmic}
	\STATE{\bfseries Input:} Effective Delay operator $A:\scrH\to\scrH$, step size $\{\tau_{n}\}_{\in\N}$, Sequences $\{\alpha_{n}\}_{n\in\N},\{\beta_{n}\}_{n\in\N}$, Initial point $h_{0}\in\scrX$.
\FOR{$n=0,1,\ldots,N$}
\STATE Obtain $h_{n}$ by running a DNL procedure; 
\IF{Stopping condition not satisfied}
\STATE		Compute 
\begin{eqnarray*}
&y_{n}&=P_{\scrX}(h_{n}-\tau_{n} A(h_{n})),\\
&z_{n+1}&=y_{n}+\tau_{n}(A(h_{n})-A(y_{n})),\\
&h_{n+1}&=(1-\alpha_{n}-\beta_{n})h_{n}+\beta_{n}z_{n+1}.
\end{eqnarray*}
\STATE Update the step size sequence 
			\begin{equation}\label{eq:adaptiveFBF}
			\tau_{n+1}=\left\{\begin{array}{cc} 
			\min\left\{\tau_{n},\frac{\mu\norm{y_{n}-h_{n}}}{\norm{A(y_{n})-A(h_{n})}}\right\} & \text{ if }A(y_{n})-A(h_{n})\neq 0,\\
			\tau_{n} & \text{ else} 
			\end{array}\right.
			\end{equation}
\ENDIF
\ENDFOR
\end{algorithmic}
\end{algorithm}
%------------------------------------
Algorithm \ref{alg:FBF} has been analyzed in \cite{DuvMeiStaVu20} in detail. In particular, we demonstrated strong convergence of the numerical scheme by proving Theorem \ref{th:FBF} below.

\begin{theorem}[Theorem 2, \cite{DuvMeiStaVu20}]
\label{th:FBF}
Suppose that Assumptions \ref{ass:1}-\ref{ass:monotone} are satisfied. Let $\{\alpha_n\}_{n\in\N}$ and $\{\beta_{n}\}_{n\in\N}$ be two real sequences in $(0,1)$, satisfying conditions 
\begin{equation}\label{eq:beta}
\{\beta_{n}\}_{n\in\N}\subseteq (b,1-\alpha_{n})\text{ for some }b>0,
\end{equation}
and 
\begin{equation}\label{eq:alpha}
\lim_{n\to\infty}\alpha_{n}=0\text{ and }\sum_{n=1}^{\infty}\alpha_{n}=\infty.
\end{equation}
%Let $\{\tau_{n}\}$ be designed by the adaptive rule \eqref{eq:adaptiveFBF}. 
Then the sequence $\{h_{n}\}$ generated by Algorithm \ref{alg:FBF} converges strongly to $h^{\ast}=\arg\min\{\norm{z}: z\in\Omega\}$.
\end{theorem}

Departing from here, our aim in this paper is to significantly extend our previous work by designing a new FBF-based inertial algorithm, which meets all the desiderata spelled out in the introduction: i) Adaptive step-sizes, ii) Weak monotonicity, and iii) strong convergence. 

To achieve a possible convergence acceleration and to meet conditions i)-iii), we include relaxation and inertial effects into our algorithm. To the best of our knowledge, this is the first available, provably convergent, relaxed-inertial splitting algorithm for computing \ac{DUE}.

The basic idea behind inertial algorithms is to use information accumulated from past iterations in order to introduce momentum. This is achieved by computing the extrapolated point $z=h+\alpha(h-h')$ in the first step of each iteration. The introduction of momentum is classical, and can be traced back to the heavy-ball method of Polyak \cite{Pol64}. We adapt momentum by injecting relaxations steps in a disciplined way to force the trajectory to converge strongly to a \ac{DUE}. The so-constructed new strongly convergent method, to be called the \emph{inertial forward-backward-forward} (IFBF) algorithm, is displayed in Algorithm \ref{alg:IFBF}.

%------------ALGORITHM-IFBF--------------------
\begin{algorithm}[t]
	\caption{IFBF for $\VI(A,\scrX)$.}
	\label{alg:IFBF}
 \begin{algorithmic}
\STATE{\bfseries Input:} Effective delay operator $A:\scrH\to\scrH$, step size $\tau_0>0$ and constants $\lambda, \mu \in (0,1)$. Sequences $\{\epsilon_{n}\}_{n\in\N},\{\beta_{n}\}_{n\in\N}$.
Initial point $h_{-1},h_{0}\in\scrX$.
\FOR{$n=0,1,\ldots,N$}
\STATE obtain $h_{n}$ via a DNL procedure
\IF{Stopping condition not satisfied}
\STATE Compute 
\begin{eqnarray*}
&w_{n}&=(1-\beta_{n})[h_{n}+\alpha_{n}(h_{n}-h_{n-1})]\\
&y_{n}&=P_{\scrX}(w_{n}-\tau_{n} A(w_{n}))\\ 
&h_{n+1}&=(1-\lambda)w_{n}+\lambda(y_{n}+\tau_{n}(A(w_{n})-A(y_{n})))
\end{eqnarray*}
\STATE Update the step size 
			\begin{equation}\label{stepsize}
			\tau_{n+1}=\left\{\begin{array}{cc} 
			\tau_{n} & \text{ if }A(w_{n})=A(y_{n})\\
			\min\{\tau_{n},\frac{\mu\norm{w_{n}-y_{n}}}{\norm{A(w_{n})-A(y_{n})}}\} & \text{else}
			\end{array}\right.
			\end{equation}
\STATE Update the inertia parameter
			\begin{equation}
			0\le \alpha_{n+1}\le 
			\left\{\begin{array}{cc} 
			\min\bigg\{\alpha, \frac{\epsilon_{n+1}}{\|h_{n+1}-h_{n}\|}\bigg\} & \text{if $h_{n+1}\ne h_{n}$}, \label{vvxvx1} \\
			\alpha &\text{ otherwise.}\\
			\end{array}\right.
			\end{equation}
\ENDIF
\ENDFOR
\end{algorithmic}
\end{algorithm}
%------------------------------------ 
In the convergence analysis of Algorithm \ref{alg:IFBF} it turns out that any positive sequences  $\{\epsilon_{n}\}_{n\in\N},\{\beta_{n}\}_{n\in\N}\subset (0,1)$, satisfying 
\begin{equation}\label{para}
\lim_{n\to\infty}\beta_n=0, \quad\sum_{n=1}^\infty \beta_n=\infty \text{  and  } \lim_{n\to\infty}\dfrac{\epsilon_{n}}{\beta_n}=0
\end{equation}
are admissible for strong convergence of the sequence of path flows $\{h_{n}\}_{n\in\N}$ generated by this Algorithm. The sequence $\{\tau_{n}\}_{n\in\N}$ has the same role as the step size $\lambda$ in the basic fixed point iteration \eqref{eq:FB}. Hence, we have to choose it small enough to ensure convergence (theoretically smaller than the reciprocal of the Lipschitz constant of the delay operator). Nevertheless we can realize IFBF without any a-priori knowledge of the Lipschitz constant by implementing the adaptive step-size policy \eqref{stepsize}. As we will see in Lemma \ref{lem:lambda}, the step size sequence $\{\tau_{n}\}_{n\in\N}$ has a limit and
\[
\lim_{n\to\infty}\tau_{n} = \tau \geq \bar{\tau} :=\min\left\{\frac{\mu}{L},\tau_{0}\right\}. 
\]
The parameter $\alpha$ can be any constant in $(0,1)$. The main theoretical result of this paper reads as follows. 

\begin{theorem}\label{th:main}
Let Assumptions \ref{ass:1}-\ref{ass:monotone} be in place. Then the sequence $\{h_{n}\}_{n\geq 0}$ generated by Algorithm \ref{alg:IFBF} converges strongly to the
minimum norm solution $h^{\ast}=\arg\min\{\norm{z}:z\in\Omega\}$.
\end{theorem}

\section{Numerical Experiments}
\label{sec:numerics}
We present preliminary computational examples of the simultaneous route-and-departure-time dynamic user equilibrium on the Nguyen network \citep{Ngu84} and the Sioux falls network. Detailed network parameters, including coordinates of nodes and link attributes, are sourced and adapted from \cite{HanEveFri19}. Given that our \ac{DUE} and DNL formulations are path-based, enumeration of paths was applied to generate the path set using the Frank-Wolfe algorithm. 

We apply Algorithm \ref{alg:FBF} and Algorithm \ref{alg:IFBF} with the embedded DNL procedure based on a time-stepping scheme discretizing the PDAE formulation described in Section \ref{sec:DNL}. We compare our method with the projected gradient algorithm \eqref{eq:FB}, as implemented in the MATLAB toolbox documented in \cite{HanEveFri19}.\footnote{The Matlab code is retrieved from \url{https://github.com/DrKeHan/DTA}.}
%%%%%
\begin{table}[h!]
\begin{center}
\begin{tabular}{|c|c|c|}
\hline
 &  Nguyen Network  & Sioux Falls   \\ \hline
No. of links & 19  & 76   \\ \hline
No. of nodes & 13  & 24  \\ \hline
No. of O-D pairs & 4  & 528   \\ \hline
No. of paths & 24 & 6,180 \\ 
\hline 
\end{tabular}
\caption{Key attributes of the test networks}
\label{tab:networks}
\end{center}
\end{table}
%%%
As remarked in \cite{HanEveFri19}, projected gradient requires strong monotonicity to ensure norm convergence, whereas all our methods are provably strongly convergent by means of Theorem \ref{th:FBF} and Theorem \ref{th:main}. All computations reported in this section were performed using MATLAB (R2018a) on a Lenovo x64 Laptop with Intel Core i5 processor with 1.6 GHz and 8GB of RAM. 

%-----------------------------------------------
\begin{figure}
    \centering
       \includegraphics[width=1\textwidth]{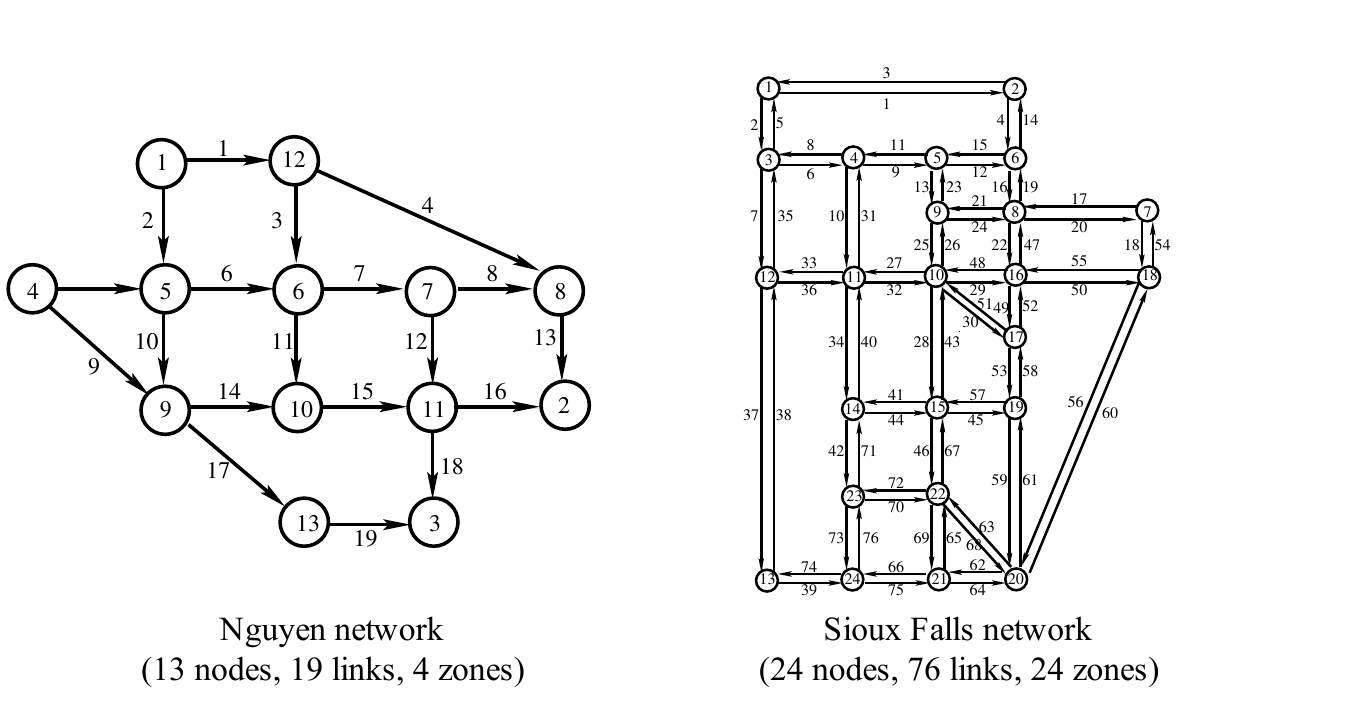}
	\caption{The Nguyen and Sioux Falls network.}
	\label{fig:networks}
	\end{figure}
%-------------------------------
\subsection{Performance of the Algorithms}
We run all three methods for fixed number of iterations and report the last iterate of the algorithm ($h^{Final}$), the corresponding effective delay operator ($A(h^{Final})$), a numerical merit function (GAP), as well as a measure for the speed of convergence. The construction of our numerical merit function follows \cite{HanEveFri19}. It is designed to measures the distance to equilibrium via the following version of a gap function 
\begin{align}\label{eq:gap}
\text{GAP}_{w}=&\max\{A_{p}(t,h^{Final}):t\in[t_{0},t_{1}],p\in\scrP_{w}\text{ s.t. }h^{Final}_{p}(t)>0\}\\
&-\min\{A_{p}(t,h^{Final}):t\in[t_{0},t_{1}],p\in\scrP_{w}\text{ s.t. }h^{Final}_{p}(t)>0\},
\nonumber
\end{align}
for all $w\in\scrW$. Hence, $\text{GAP}_{w}$ represents the range of travel costs experienced by all drivers in the given O-D pair $w\in\scrW$. In fact, it is clear that $\text{GAP}_{w}\geq 0$, and in an exact \ac{DUE}, the gap should be zero for all O-D pairs, justifying the interpretation of GAP as a numerical merit function.\\
%%%%O-D gap Nguyen
\begin{figure}[h!]
\begin{center}$
\begin{array}{lll}
\includegraphics[width=40mm]{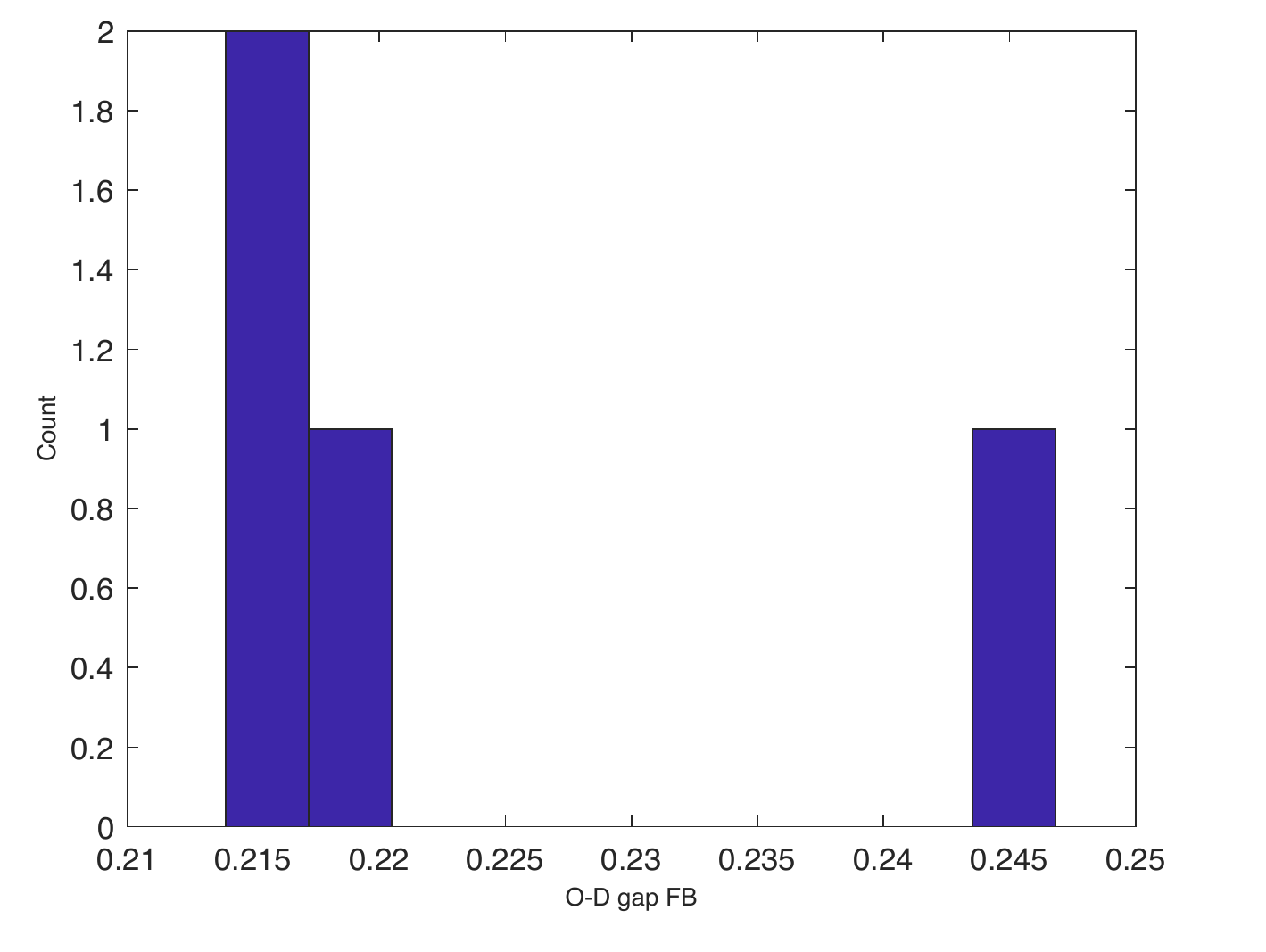}&
\includegraphics[width=40mm]{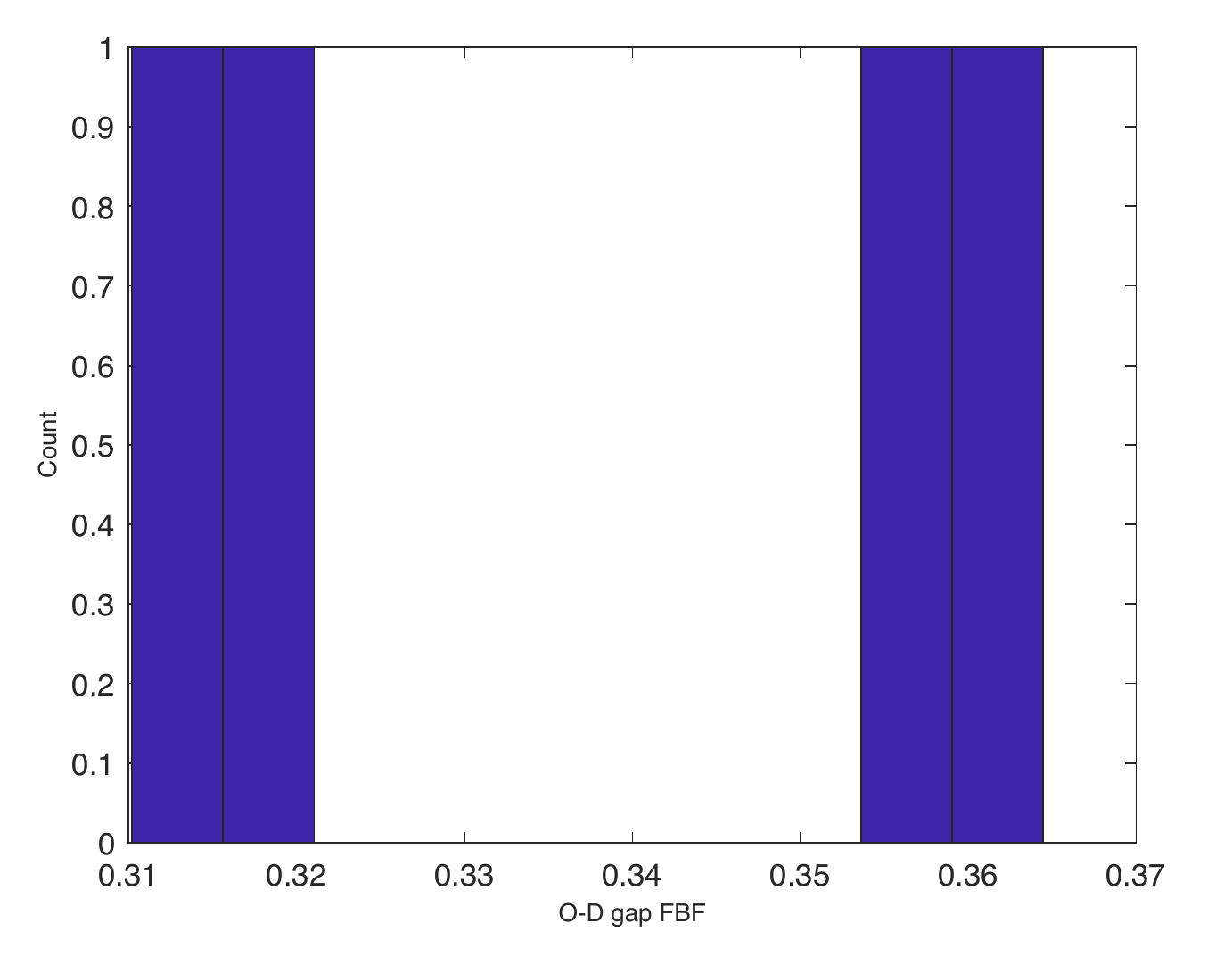}&
\includegraphics[width=40mm]{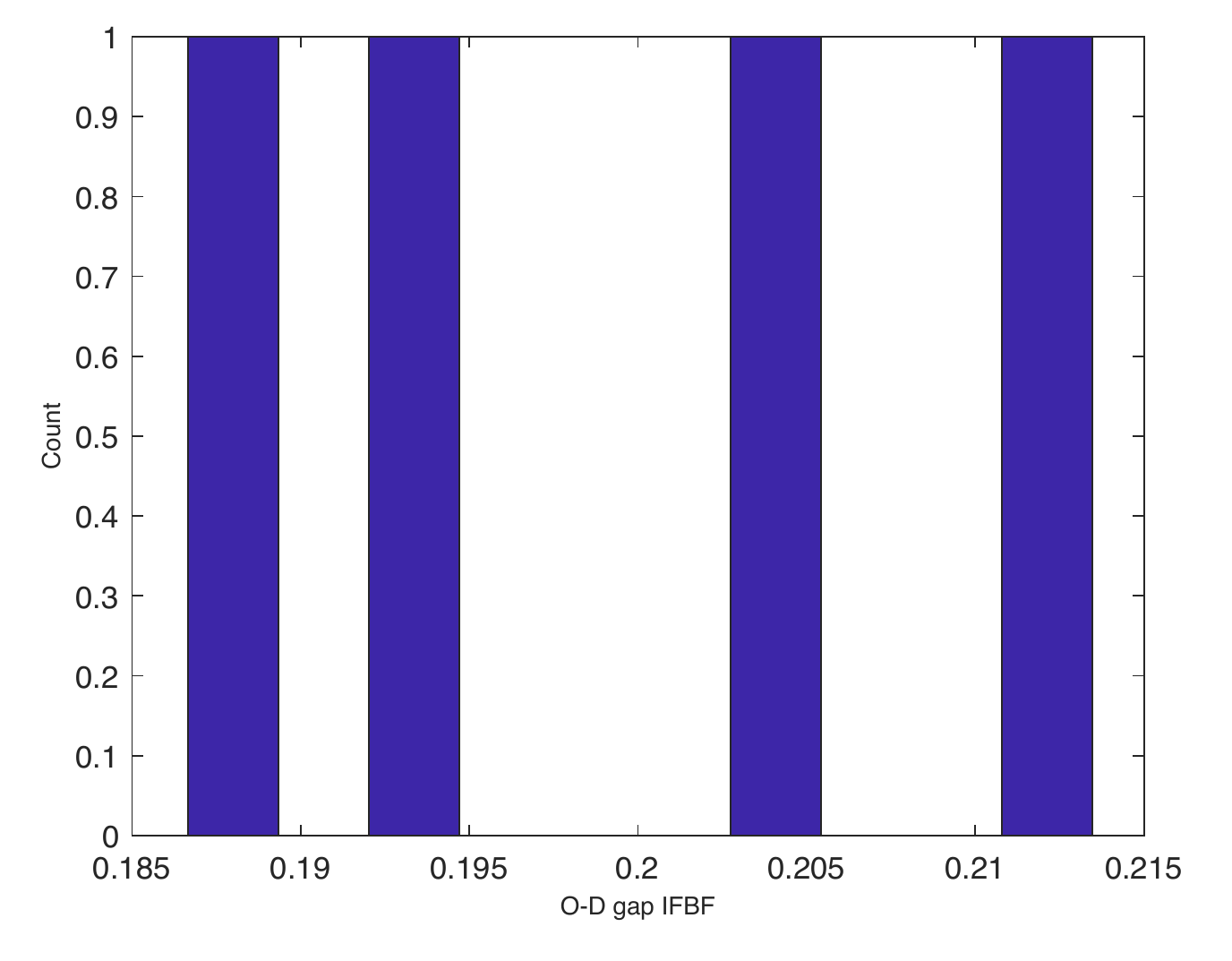}
\end{array}$
\end{center}

\begin{center}$
\begin{array}{lll}
\includegraphics[width=40mm]{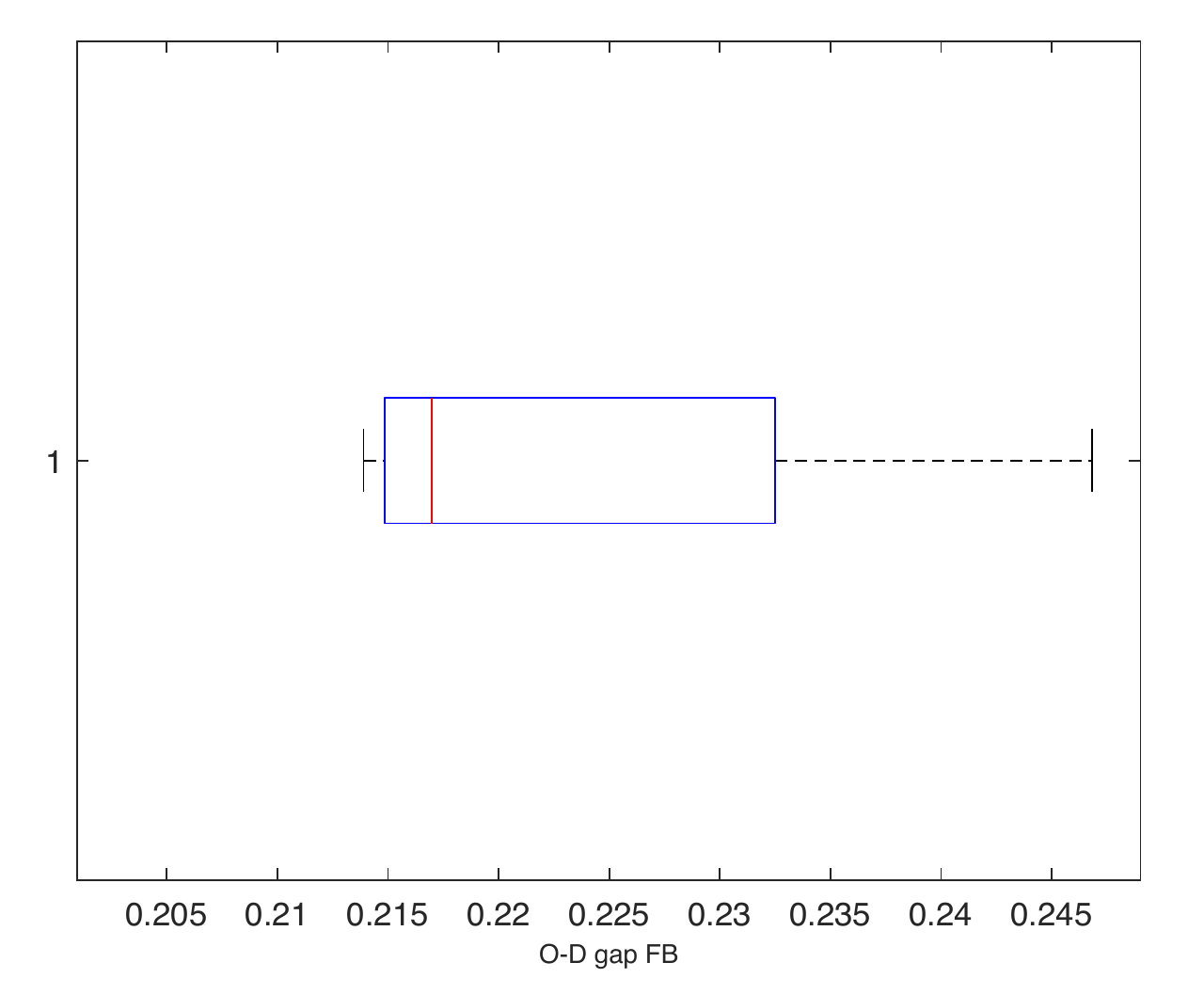}&
\includegraphics[width=40mm]{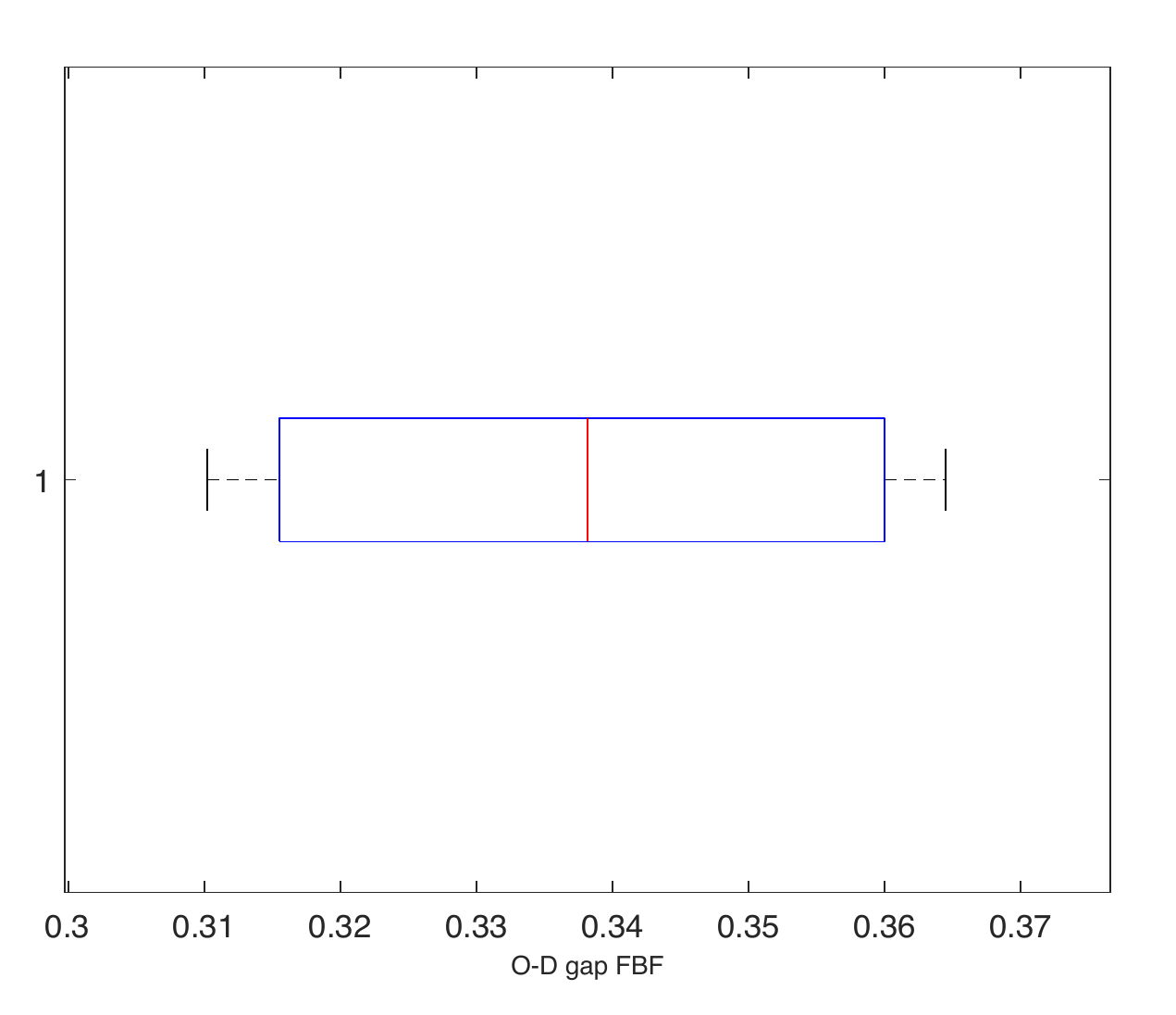}&
\includegraphics[width=40mm]{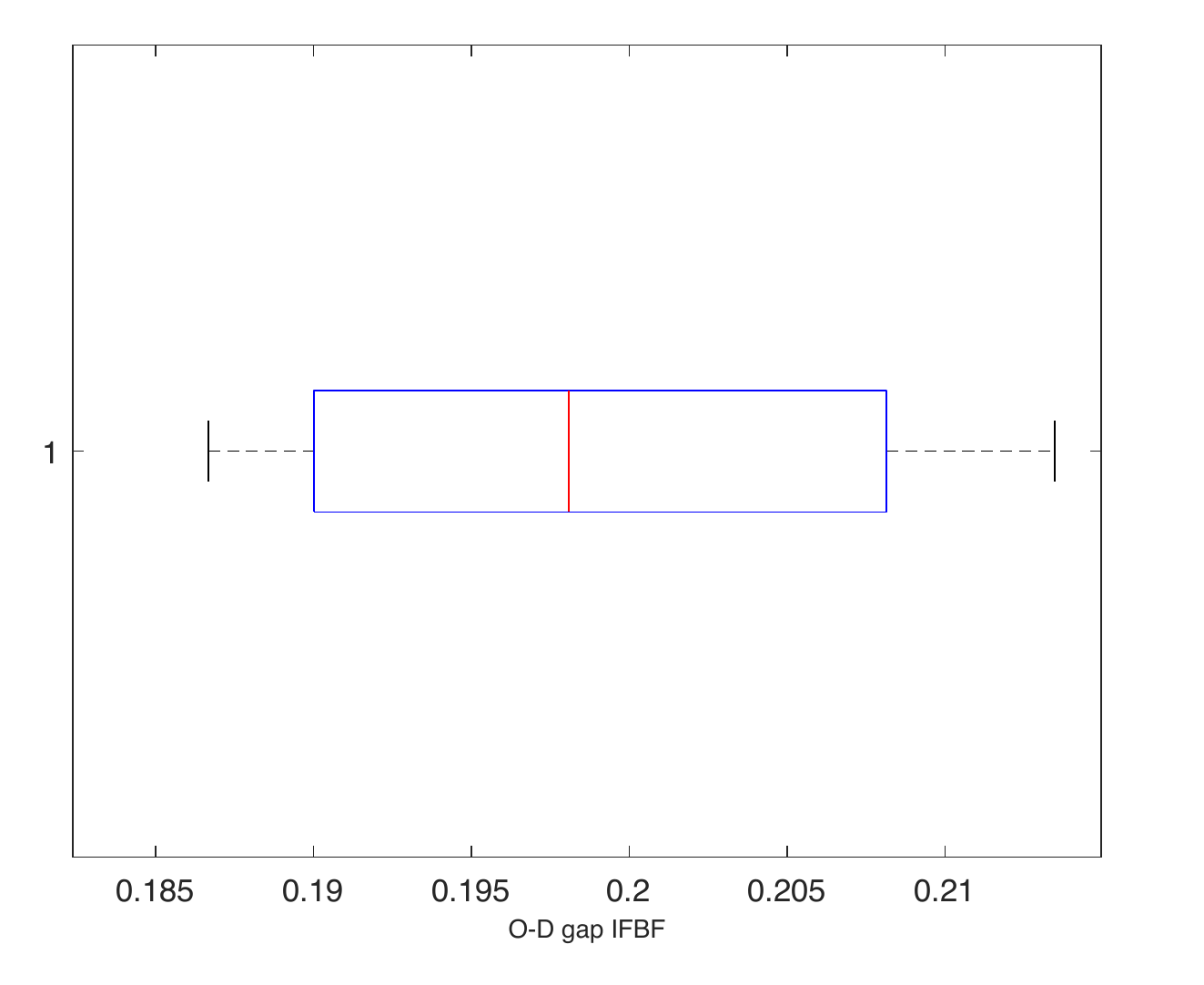}
\end{array}$
\end{center}
\caption{Distributions of O-D gaps to the \ac{DUE} solutions in the Nguyen network, calculated according to \eqref{eq:gap}.}
\label{fig:OD_Nguyen}
\end{figure}
%%%%%%%%%%%%%
\begin{table}[h!]
\centering
\begin{tabular}{lll}
                        & Nguyen & Sioux Fall   \\
                        \hline
dt                      & 70     & 100    \\
max. Iterations    & 200    & 100    \\
$\alpha$ & 0.7 & 0.7 \\
$\mu$                  & 0.5    & 0.5      \\
$\lambda$              & 0.5    & 0.2    \\
\hline  
\end{tabular}
\caption{Global Parameters for Algorithms \ref{alg:FB}, \ref{alg:FBF} and \ref{alg:IFBF}.}
\label{tab:globalpara}
\end{table}
Table \ref{tab:globalpara} contains a list of the global parameters employed in our numerical experiments. Here $dt$ is regulating the mesh-size of the time grid in the numerical solution of the \ac{DNL}. max. Iterations is the total number of iterations we let all three algorithms run on each test instance, and $\lambda$ is the relaxation parameter in Algorithm \ref{alg:IFBF}. The construction of local parameters has been done in a simple way, without involving extensive search over the parameter space which would very likely improve the reported results. The FB Algorithm \ref{alg:FB} has been implemented as in \cite{HanEveFri19}, using the constant step size $\tau=10$ on the Nguyen, and $\tau=2$ on the Sioux falls test network. FBF and IFBF (Algorithms \ref{alg:FBF} and \ref{alg:IFBF}) are implemented with the adaptive step-size \eqref{eq:adaptiveFBF} and \eqref{stepsize}, respectively. The relaxation and inertial sequences employed in FBF and IFBF are reported in Table \ref{tab:Specific}.
%%%%%%%%%%
\begin{table}[h!]
\centering
\begin{tabular}{|l|l|l|l|l|}
\hline
\multirow{2}{*}{} & \multicolumn{2}{c|}{FBF}               & \multicolumn{2}{c|}{IFBF}                             \\ \cline{2-5} 
                  & Nguyen                & Siuox Falls    & Nguyen                           & Sioux Falls        \\ \hline
$\alpha_{n}$      & $(1+n)^{-0.9}$          & $(1+n)^{-0.9}$ & N.N.                             & N.N.               \\ \hline
$\beta_{n}$       & $0.7-0.7(1+n)^{-0.7}$ & $0.5 - 0.5*(1+n)^{-0.4}$  & $(10+n)^{-2}$                    & $\frac{10}{10n+1}$ \\ \hline
$\epsilon_{n}$    & N.N.                  & N.N.           & $\left(\frac{2}{1+n}\right)^{5}$ & $(0.1+n)^{-1.1}$   \\ \hline
\end{tabular}
\caption{List of method-specific parameters. N.N. stands for ``not needed''. }
\label{tab:Specific}
\end{table}
%%%%%%%%%%%%%%%%%%%%%%%%%%%%
Figure \ref{fig:OD_Nguyen} shows the distribution of the values of our merit function \eqref{eq:gap} on the Nguyen network, and Figure \ref{fig:OD_Sioux} displays the same statistic for the Sioux falls network.

We see that in all our experiments the distribution of the O-D gaps is concentrated around 0.2 across all networks for FB and IFBF. This suggests that these algorithms preform similarly in terms of producing approximate equilibrium solutions. The decisive advantage of IFBF is however that it is guaranteed to converge strongly to the minimum norm solution, without requiring strict monotonicity of the delay operator. 
%%%%%%%%O-D gap Sioux Falls
\begin{figure}[h!]
\begin{center}$
\begin{array}{lll}
\includegraphics[width=40mm]{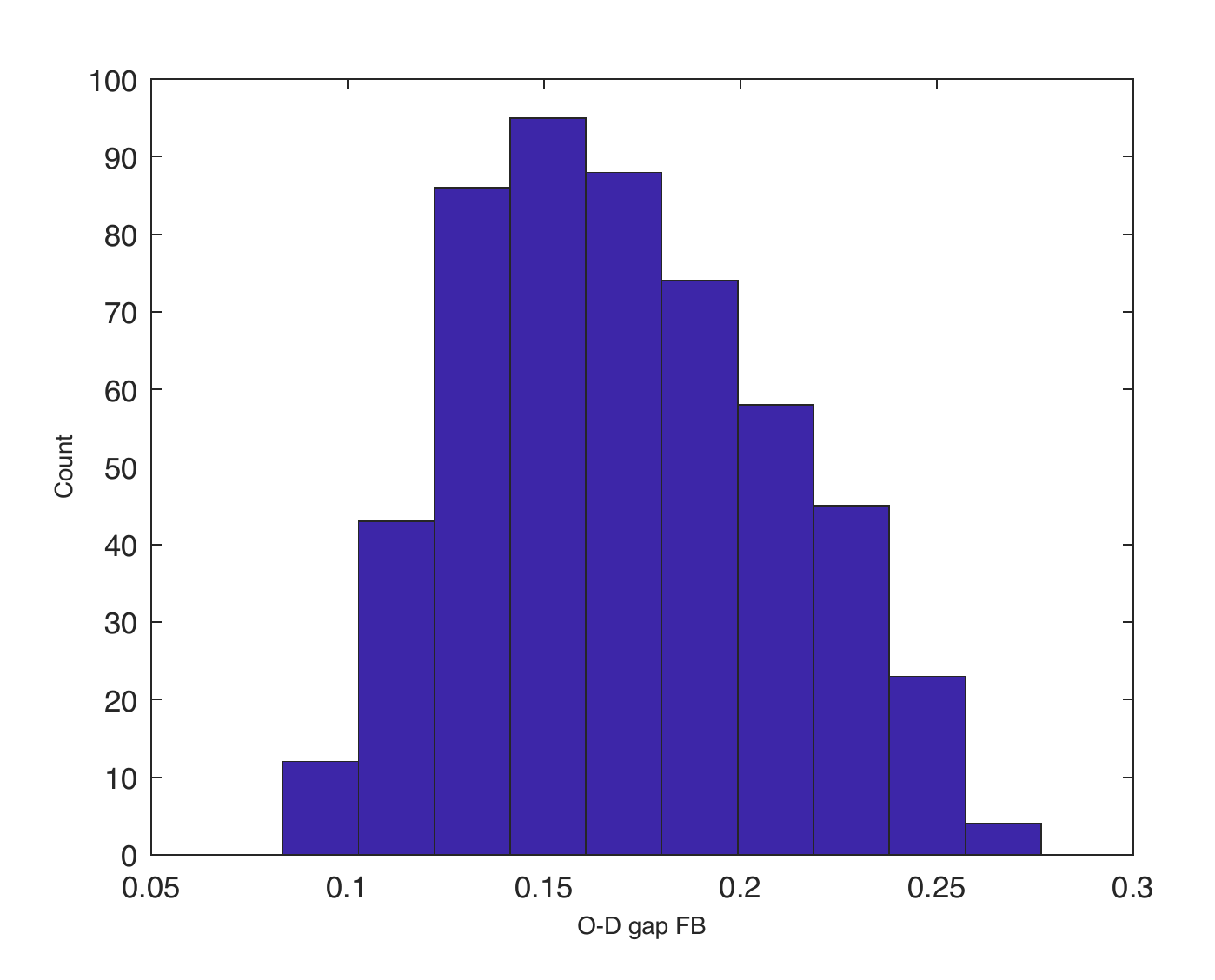}&
\includegraphics[width=40mm]{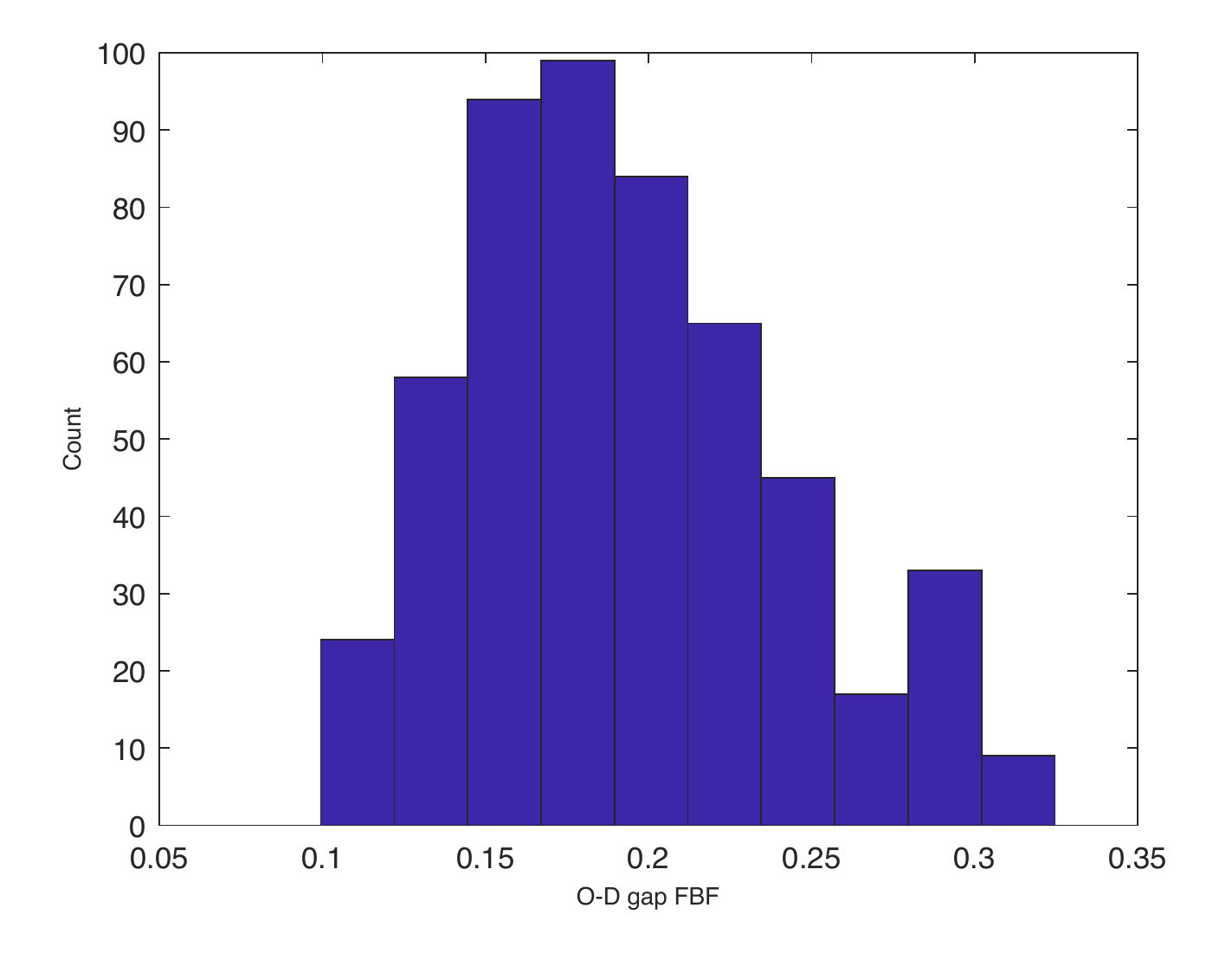}&
\includegraphics[width=40mm]{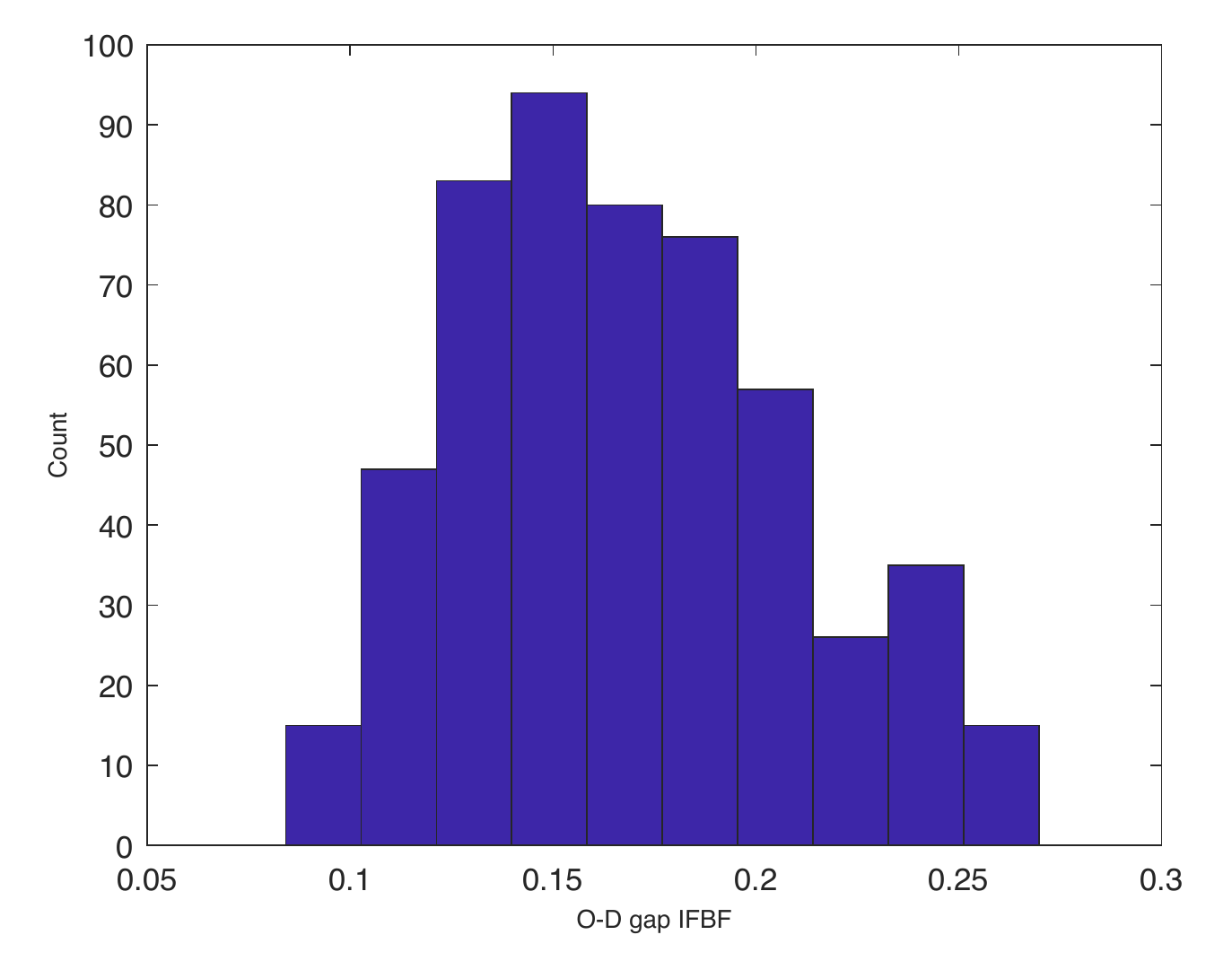}
\end{array}$
\end{center}

\begin{center}$
\begin{array}{lll}
\includegraphics[width=40mm]{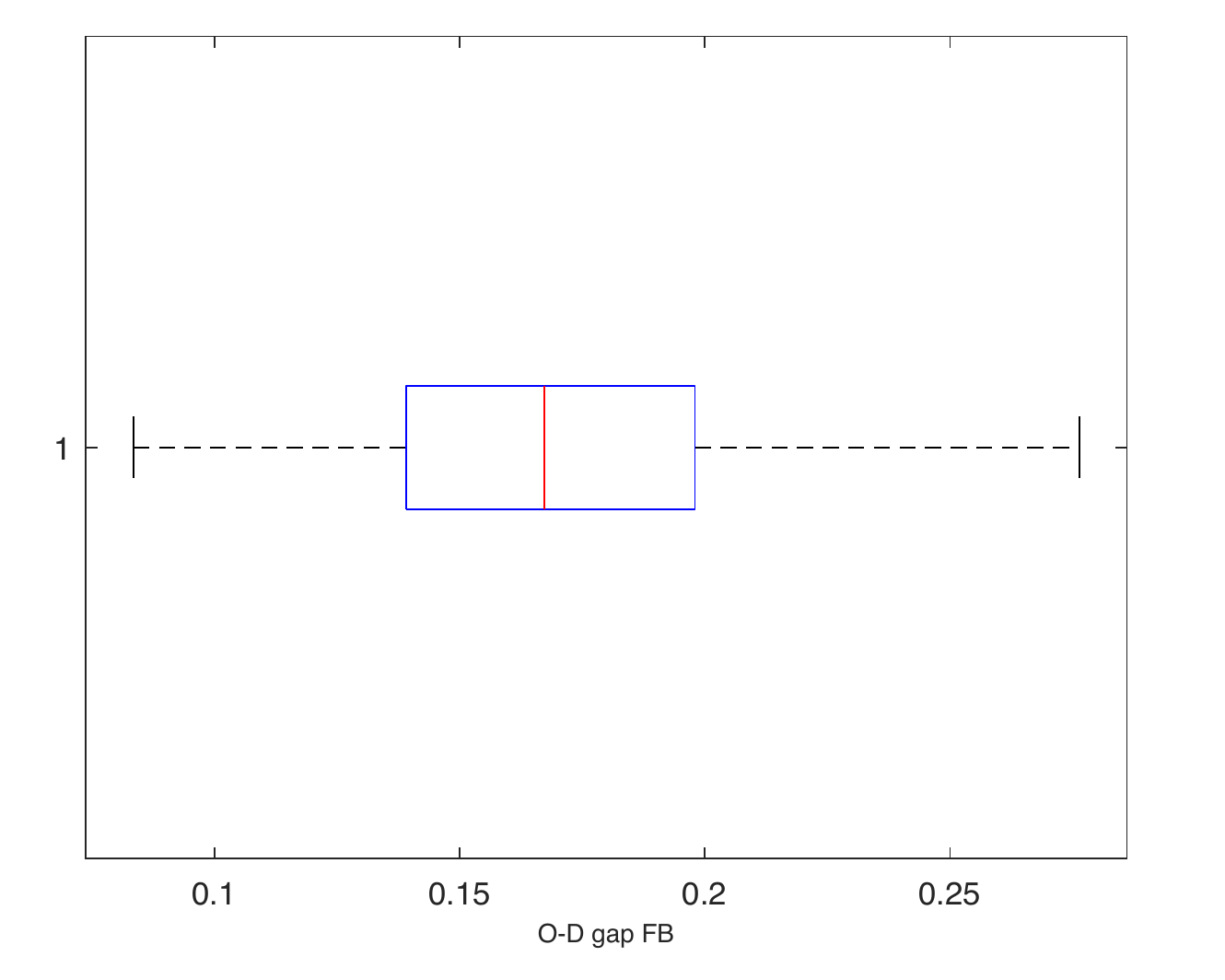}&
\includegraphics[width=40mm]{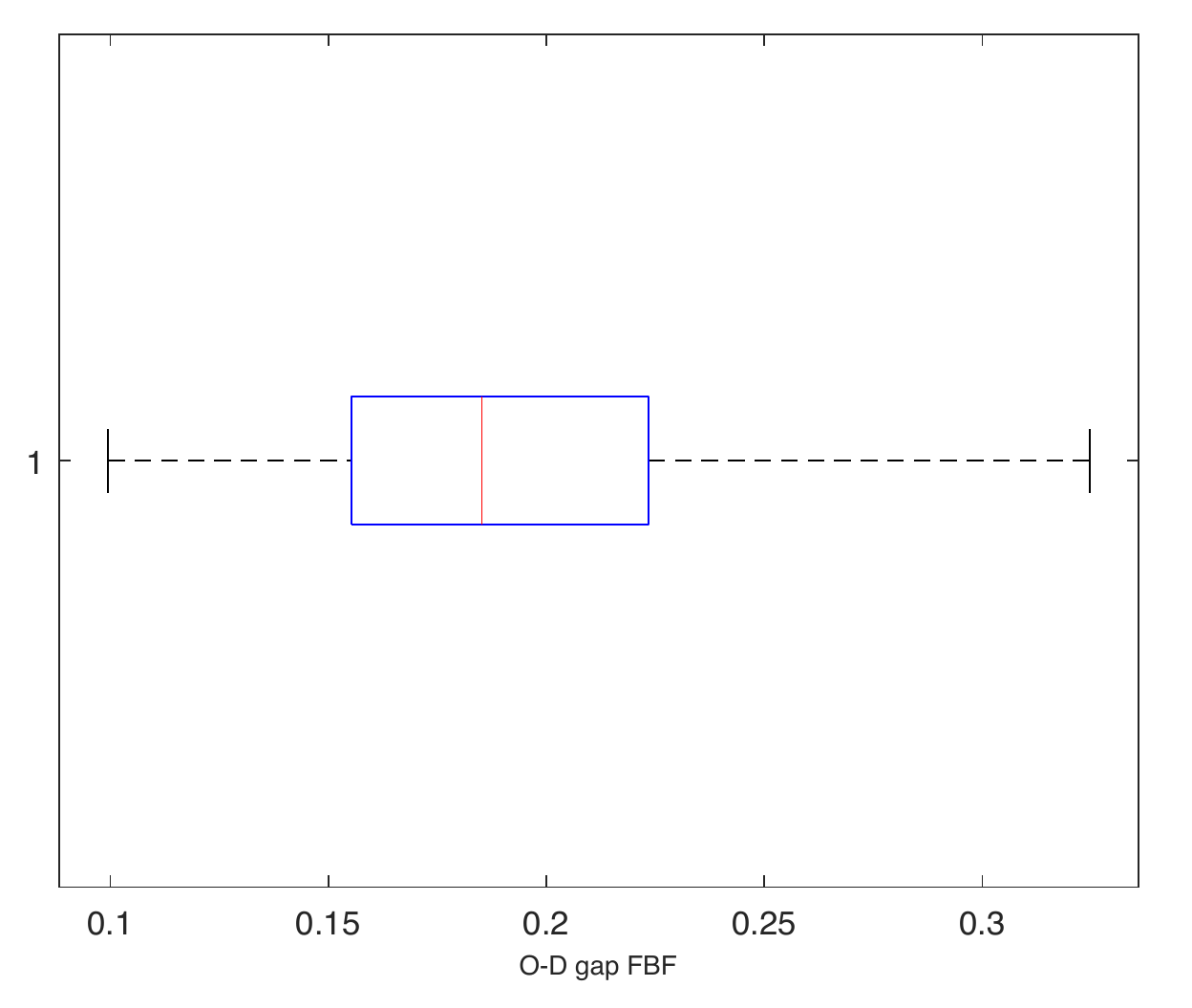}&
\includegraphics[width=40mm]{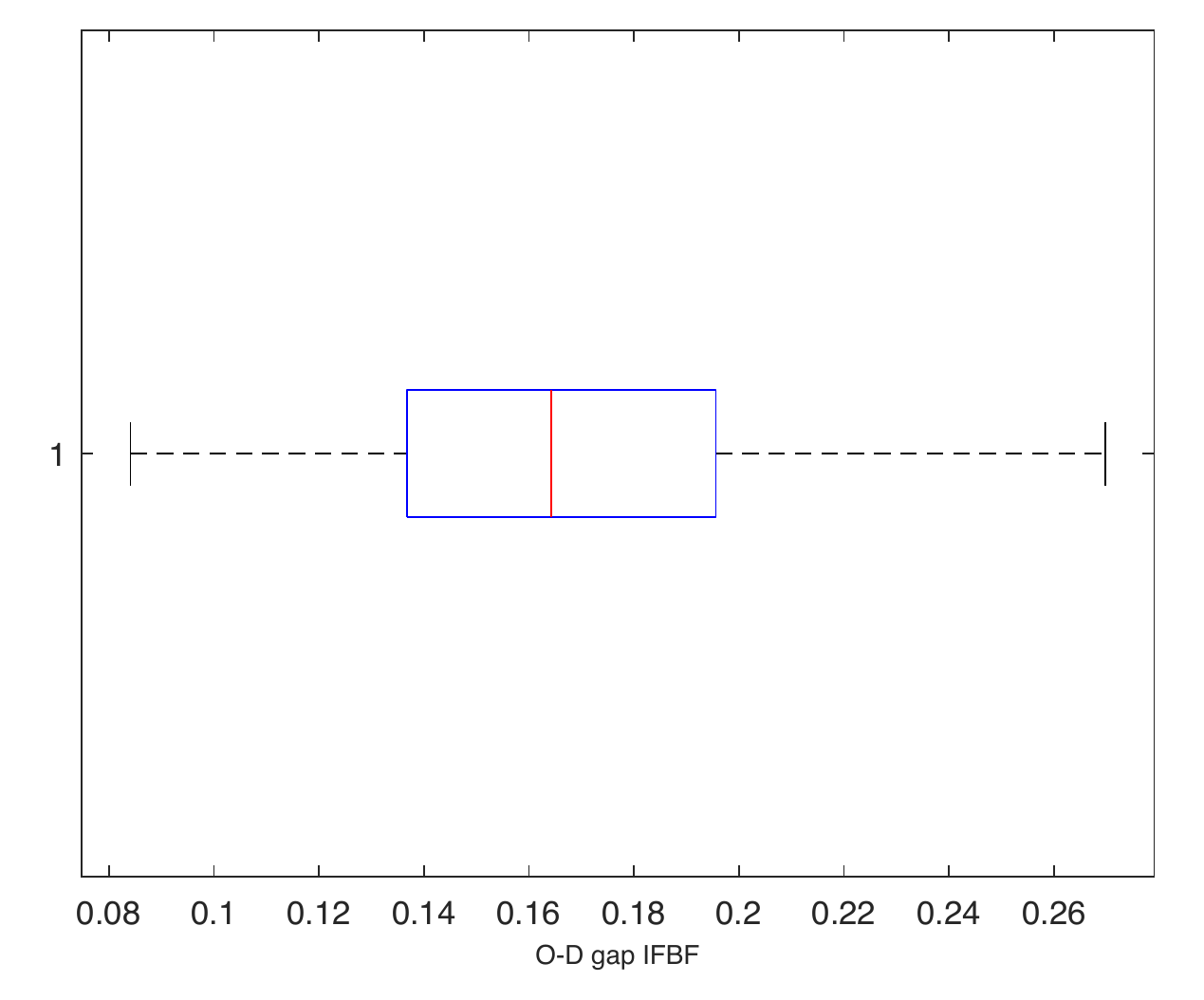}
\end{array}$
\end{center}\caption{Distributions of O-D gaps to the \ac{DUE} solutions in the Sioux falls network, calculated according to \eqref{eq:gap}.}
\label{fig:OD_Sioux}
\end{figure}
%%%%%%%%%%%%%%%%%%%%%
Figure \ref{fig:PP} displays the path departure rates as well as the corresponding effective path delays on randomly selected paths in the considered test networks. We see from these plots that the path departure rates peak out around the minima of the effective delay, which reflects equilibrium behavior on the routes. 
%%%%%Path Profile Nguyen%%%%%
\begin{figure}[h!]
    \centering
    \begin{subfigure}[b]{0.50\textwidth}
\includegraphics[width=65mm]{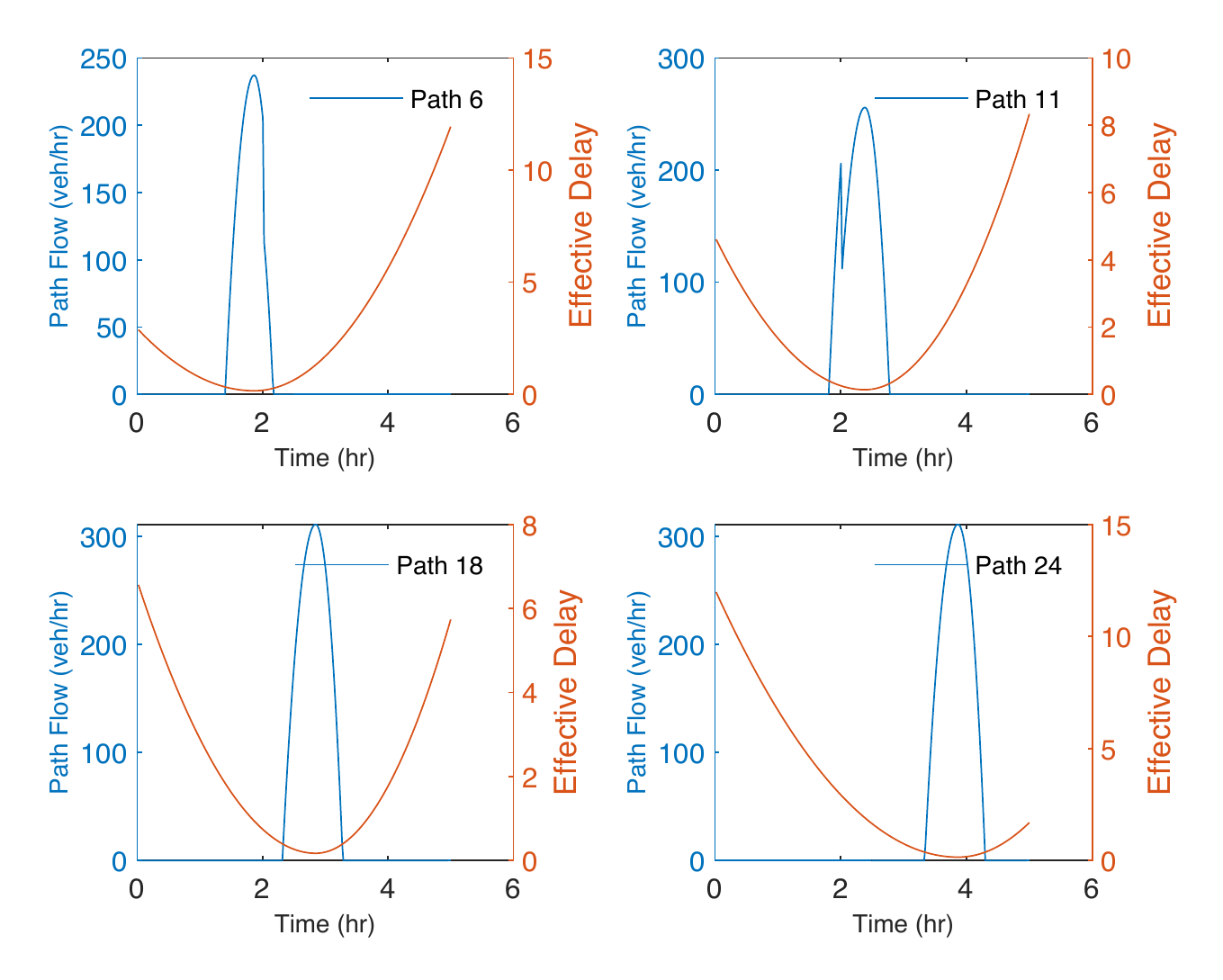}
\caption{FB Nguyen}
\end{subfigure}\hfill
\begin{subfigure}[b]{0.5\textwidth}
\includegraphics[width=67mm]{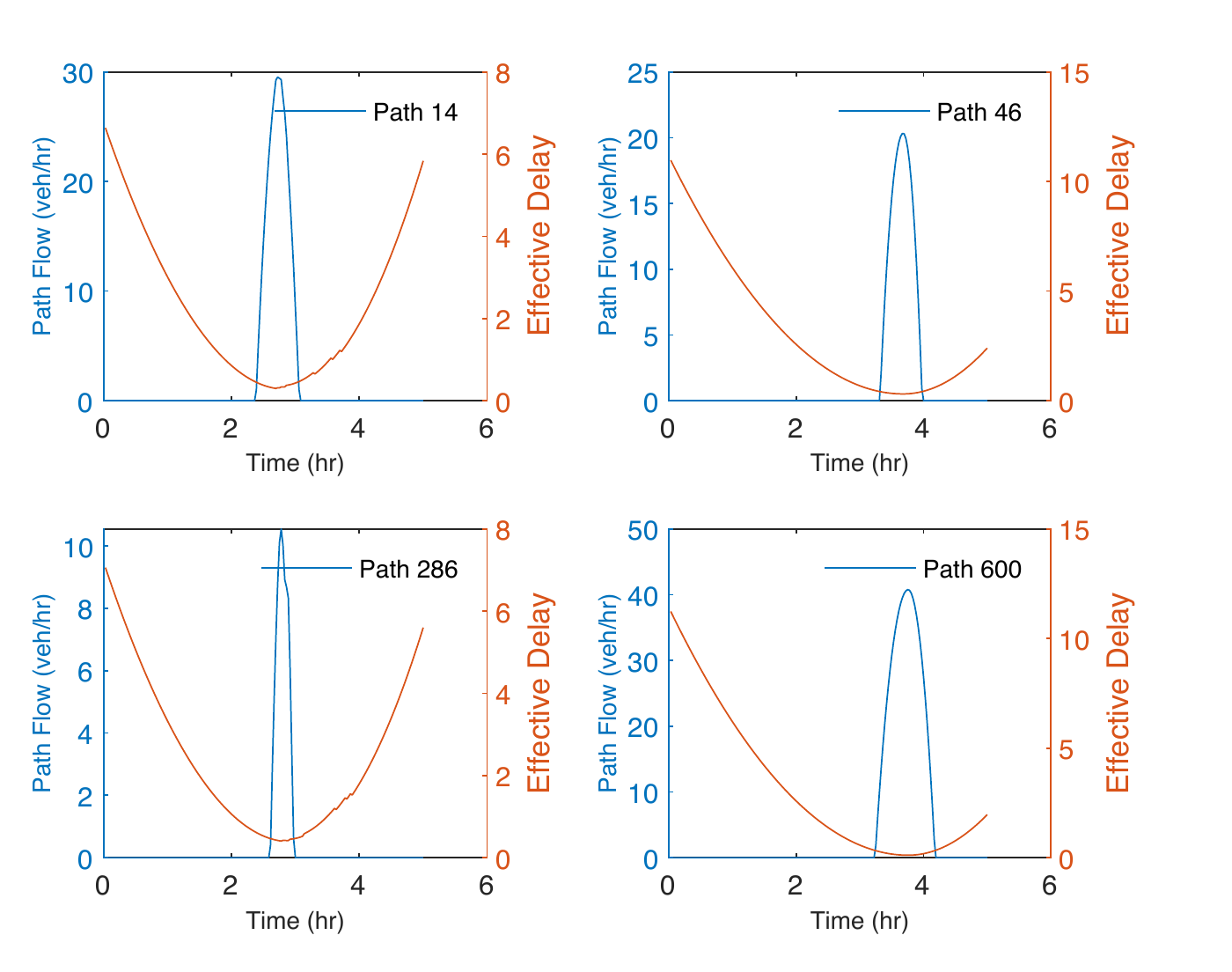}
\caption{FB Sioux}
\end{subfigure}\hfill
%%%%%%%%%%%%%%%%%%%%%%%%%%%%%%
\begin{subfigure}[b]{0.5\textwidth}
\includegraphics[width=65mm]{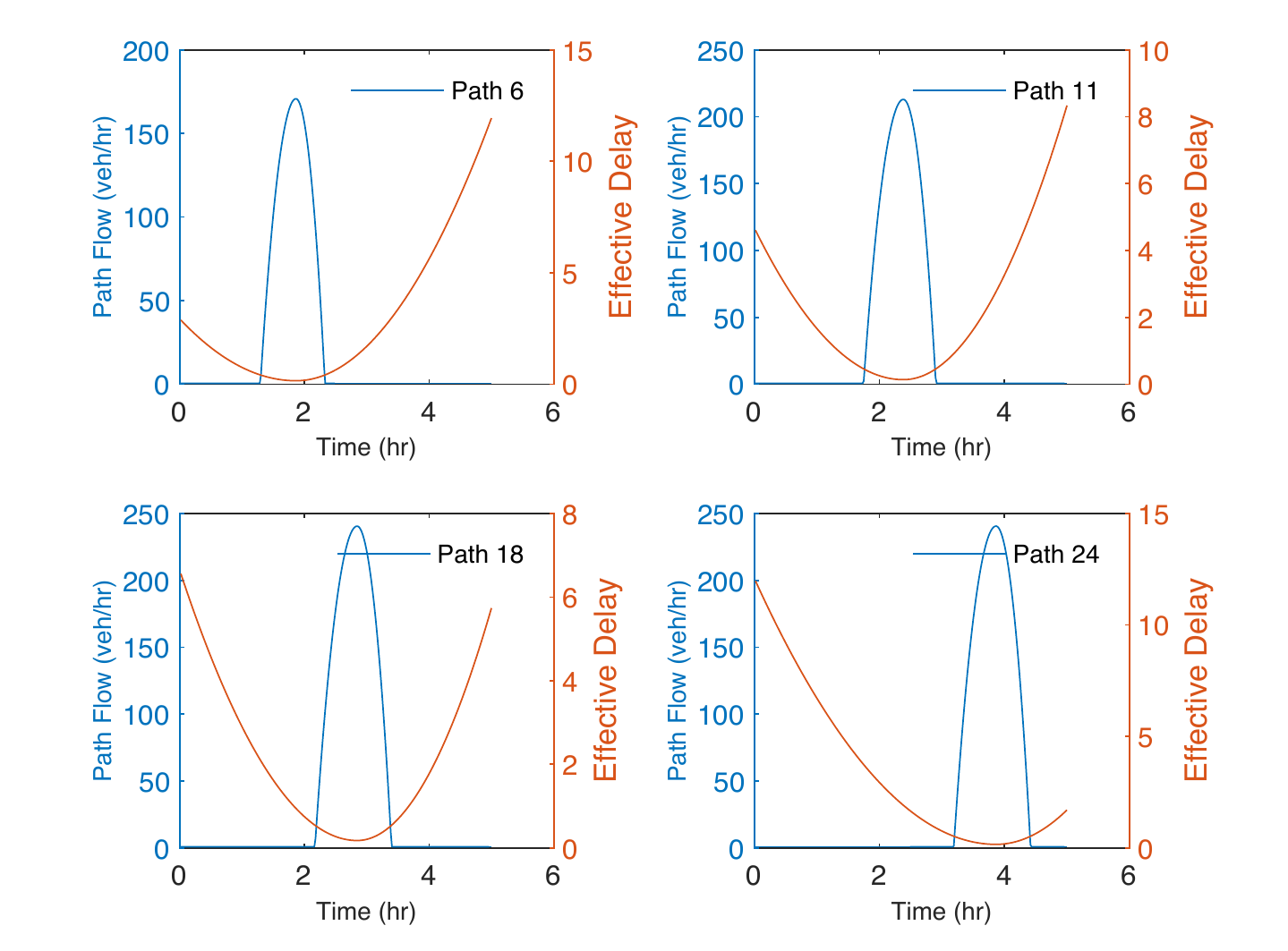}
\caption{FBF Nguyen}
\end{subfigure}\hfill
\begin{subfigure}[b]{0.5\textwidth}
\includegraphics[width=63mm]{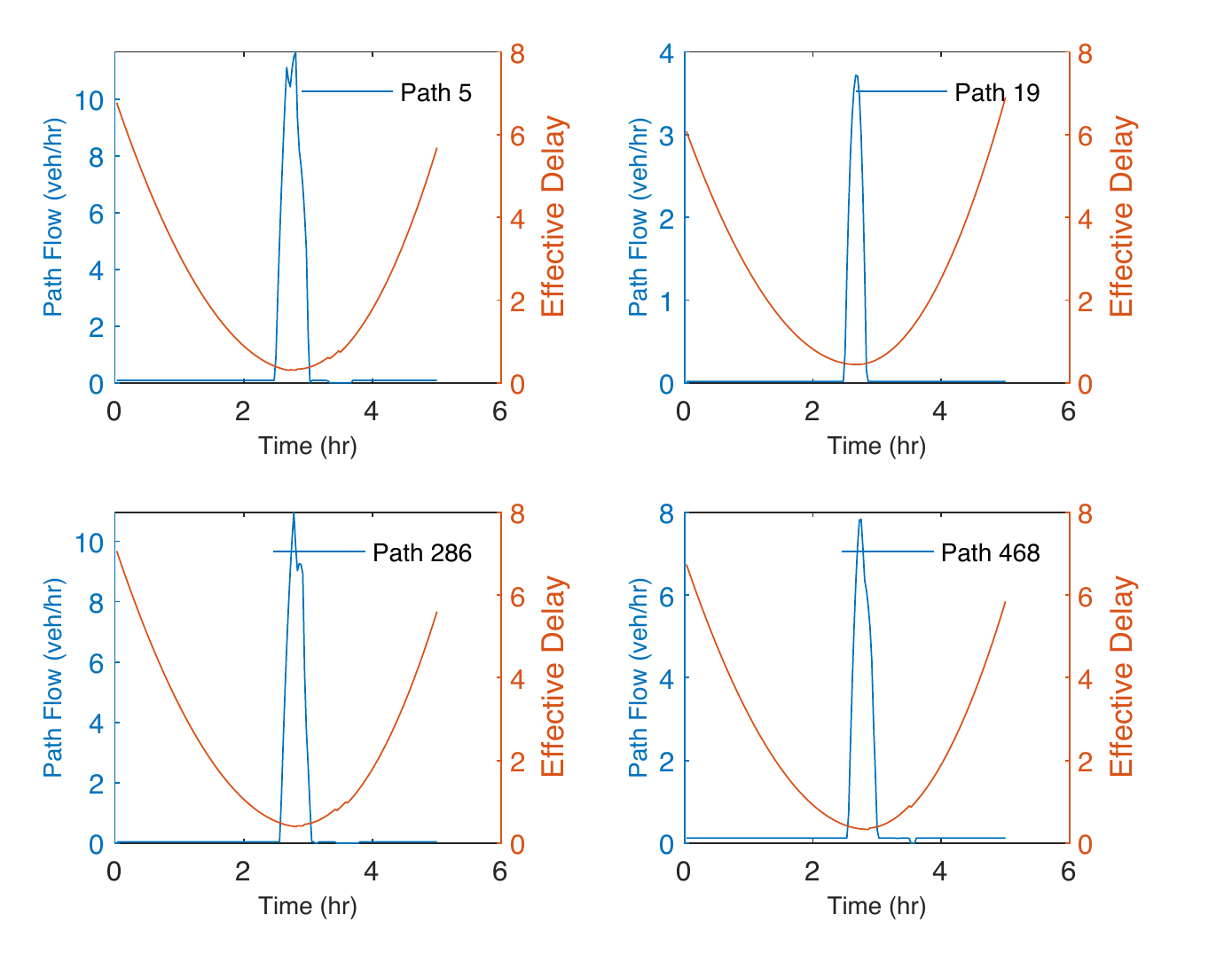}
\caption{FBF Sioux}
\end{subfigure}\hfill
%%%%%%%%%%%%%%%%%%%%%%%%%
\begin{subfigure}[b]{0.5\textwidth}
\includegraphics[width=65mm]{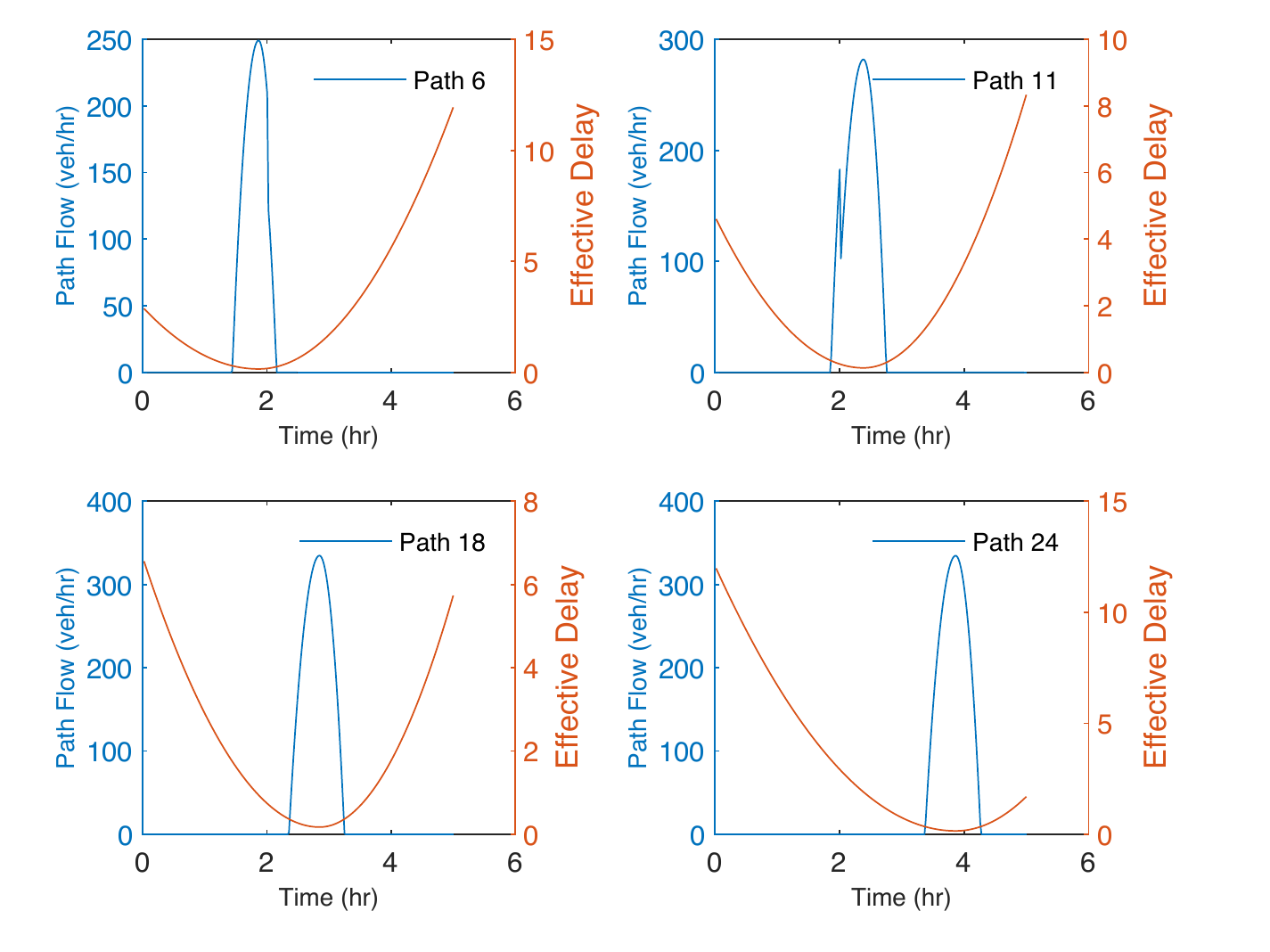}
\caption{IFBF Nguyen}
\end{subfigure}\hfill
\begin{subfigure}[b]{0.5\textwidth}
\includegraphics[width=63mm]{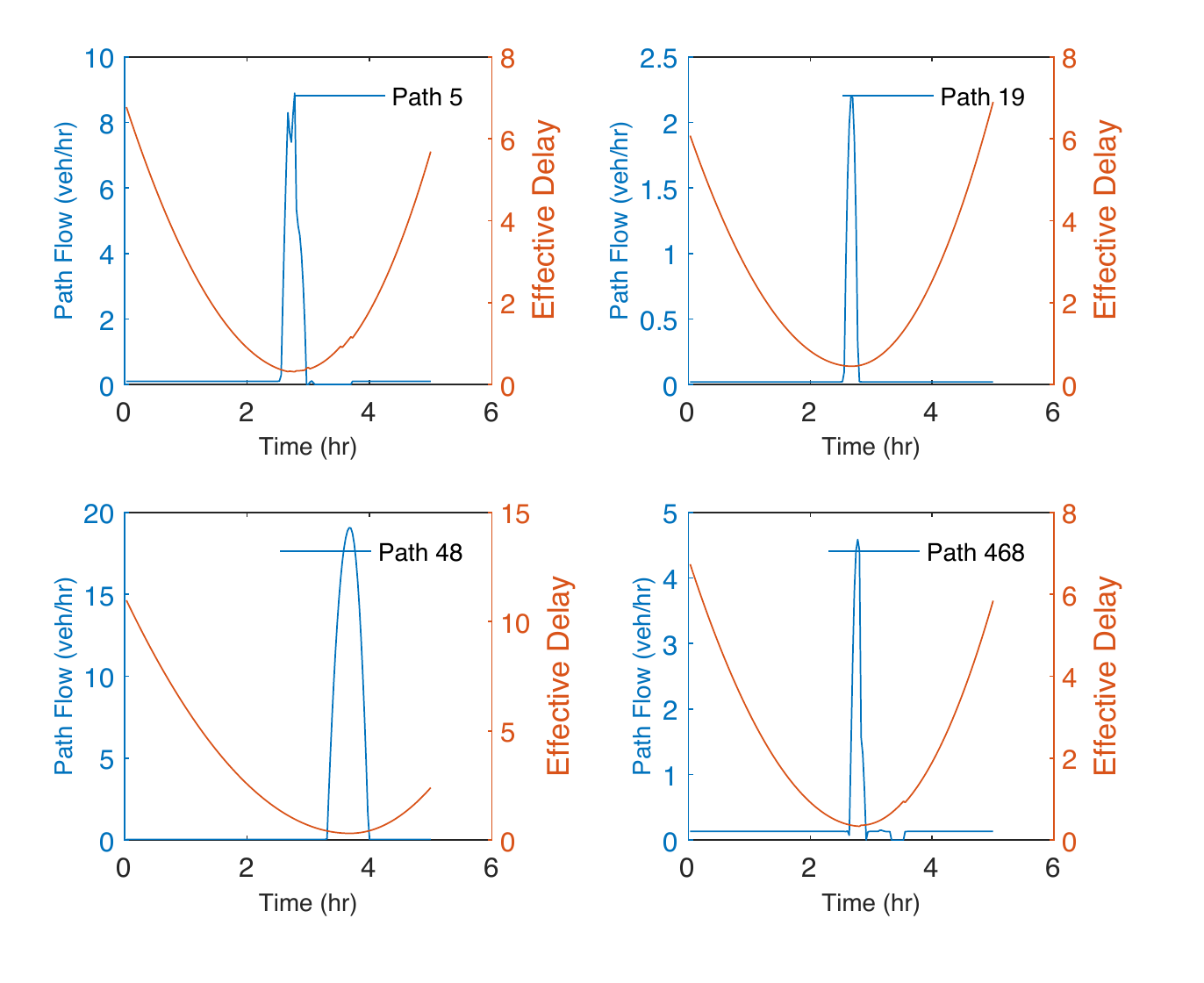}
\caption{IFBF Sioux}
\end{subfigure}
\caption{Path departure rates and corresponding effective path delays for selected paths in the \ac{DUE} solutions on the Nguyen network.}
\label{fig:PP}
\end{figure}

%%%%%%%
We finally display a figure which gives some indication on the relative speed of convergence of each of the tested algorithms. We compute for each method the ``relative energy'' sequence 
\begin{equation}\label{eq:energy}
e_{n} = \frac{\norm{h_{n+1}-h_{n}}}{\norm{h_{n}}}
\end{equation}
which measures the decay of energy of the path departure rates generated by the algorithms. \cite{HanEveFri19} call this the relative gap, and we follow this terminology in the labeling of the figures. For all our methods this sequence must converge to 0, and we can consider one method faster than the other if the rate of convergence of the energy sequence dominates the other. Figure \ref{fig:epsilon} shows the evolution of the relative energy sequences for each method. 
%%%%%%%%Epsilon%%%%%
\begin{figure}[h!]
    \centering
    \begin{subfigure}[b]{0.5\textwidth}
	   \includegraphics[width=58mm]{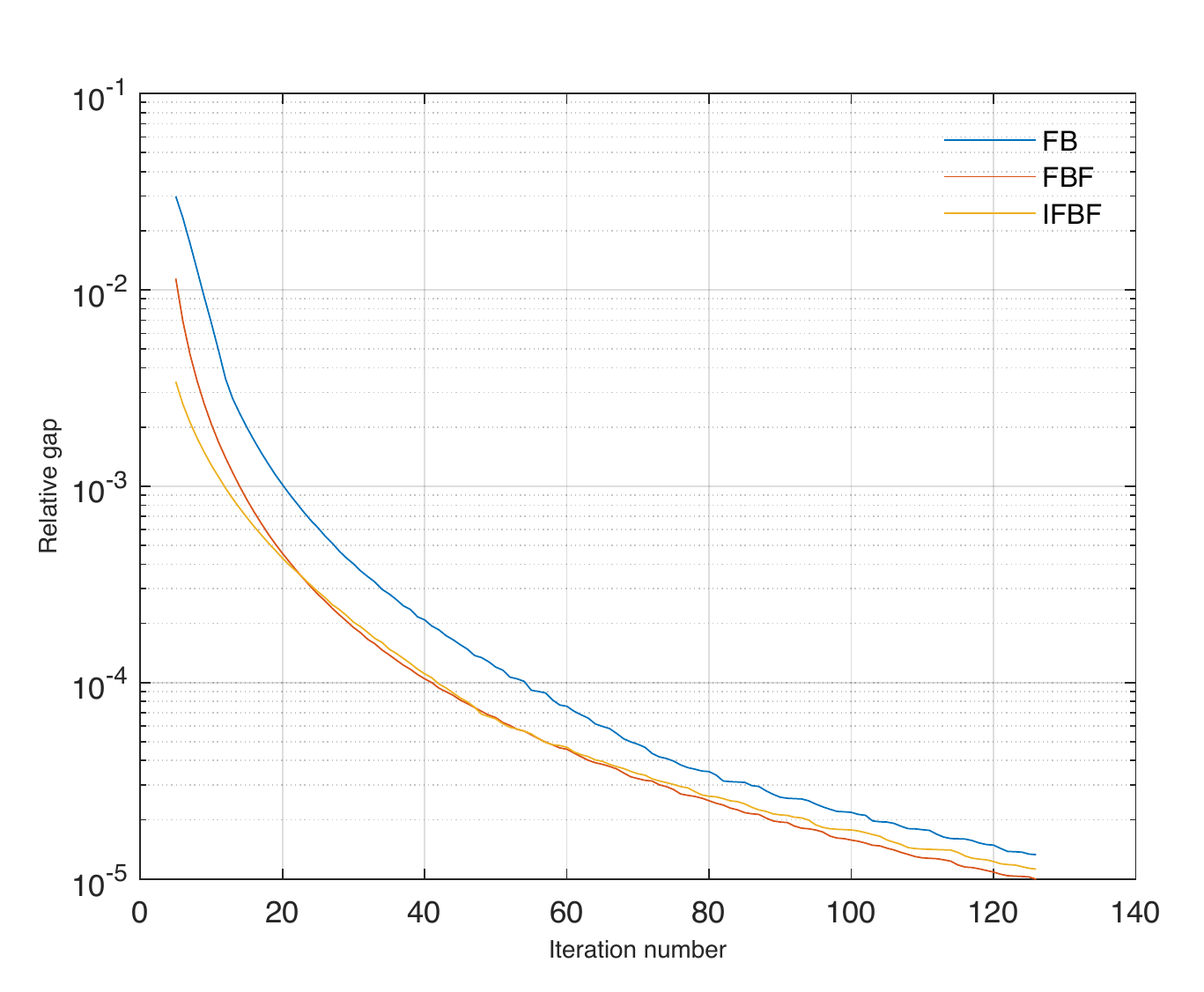}
\caption{Nguyen network}
\end{subfigure}\hfill
    \begin{subfigure}[b]{0.5\textwidth}
	   \includegraphics[width=60mm]{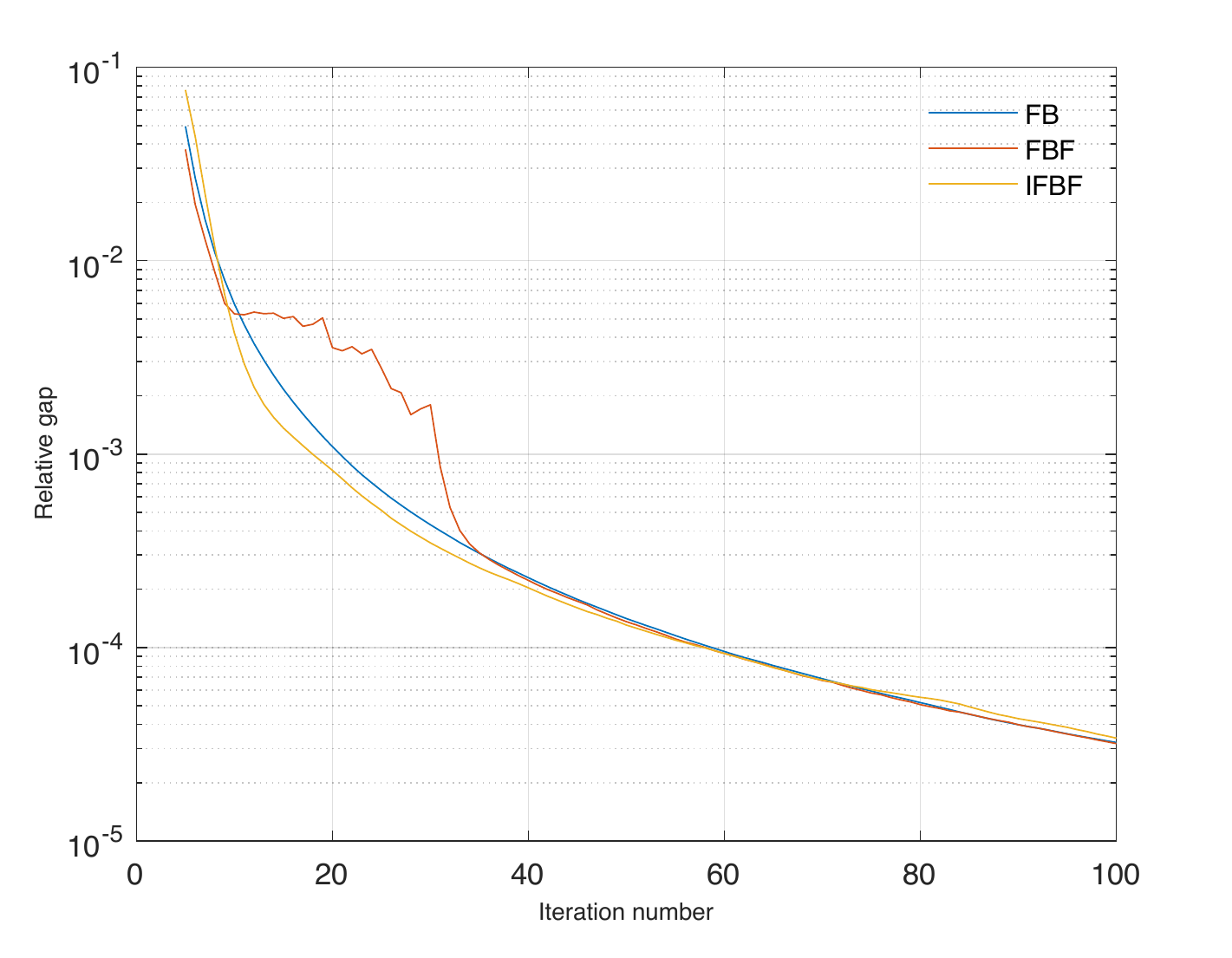}
\caption{Sioux Falls network}
\end{subfigure}
\caption{The relative energy \eqref{eq:energy} on semi-log scale for both test networks.}
\label{fig:epsilon}
\end{figure}
%%%
We see that already non-optimized step size parameters lead to some acceleration in the IFBF scheme when compared to other solvers. 

\section{Convergence Analysis}

\subsection{Preliminaries}
\label{sec:prelims}
The purpose of this section is to collect some standard concepts from real Hilbert spaces. Throughout this section we let $\scrH$ be a real Hilbert space with scalar product $\inner{\cdot,\cdot}$ and associated norm $\norm{\cdot}$. 

\begin{definition}[Convergence in Hilbert spaces]
\label{def:convergence}
A sequence of points $\{x_{n}\}_{n\in\N}$ in a Hilbert space $\scrH$ converges weakly to a point $x\in\scrH$, denoted by $x_{n}\wlim x$, if 
\[
\lim_{n\to\infty}\inner{x_{n},y}=\inner{x,y}
\]
for all test vectors $y\in\scrH$. The sequence $\{x_{n}\}$ converges strongly to $x$ if 
\[
\lim_{n\to\infty}\norm{x_{n}-x}=0.
\]
\end{definition}

In order to prove our main convergence results, we need the following standard facts. 
For all $x,y\in \scrH$ and $\alpha\in \mathbb{R}$, we have
\begin{equation}\label{xx22}
\|x+y\|^2\le \|x\|^2+2\langle y,x+y\rangle.
\end{equation}

\begin{lemma}\label{lem:projection}
Let $C$ be a nonempty closed convex set in $\scrH$ and $x\in\scrH$ arbitrary. Then 
\begin{itemize}
\item[(i)] $\inner{P_{C}(x)-x,y-P_{C}(x)}\geq 0$ for all  $y\in C;$
\item[(ii)] $\norm{P_{C}(x)-y}^{2}\leq\norm{x-y}^{2}-\norm{x-P_{C}(x)}^{2}$ for all $y\in C$. 
\end{itemize}
\end{lemma}

\begin{definition}
Let $A:\scrH\to\scrH$ be an operator. The operator $A$ is called 
\begin{enumerate}
\item $L$-Lipschitz continuous with $L>0$ on $\scrX\subseteq\scrH$ if 
\[
\norm{A(x)-A(y)}\leq L\norm{x-y}\qquad\forall x,y\in\scrX.
\]
\item pseudo-monotone on $\scrX\subseteq\scrH$ if
\begin{equation}\label{eq:PM}
\inner{A(x),y-x}\geq 0\Rightarrow \inner{A(y),y -x}\geq 0\qquad\forall x,y\in\scrX.
\end{equation}

\item sequentially weakly continuous if $x_{n}\wlim x$ then $A(x_{n}) \wlim A(x)$. 
\end{enumerate}
\end{definition}
The next classical fact shows that solutions of $\VI(\scrX,A)$ defined in terms of pseudo-monotone operators can be determined via ``weak formulation''. 

\begin{lemma}\label{002}
\citep{CotYao92}
	If $A: \scrH  \to \scrH $ is pseudo-monotone
	and continuous, then $x^*$ is a solution of $VI(\scrX, A)$ if and only if
	\[
	\inner{A(x), x - x^{\ast}}\geq 0 \  \ \forall x \in  \scrX.
\]
\end{lemma}
The next technical results are basic convergence guarantees for real-valued sequences.

\begin{lemma}\label{lem:summation}\citep{Xu02}
Let $\{a_{n}\}_{n\in\N}$ be a sequence of nonnegative real numbers, and $\{\beta_{n}\}_{n\in\N}$ be a sequence in $(0,1)$ such that $\sum_{n}\beta_{n}=\infty$. Suppose that $\{b_{n}\}_{n\in\N}$ is a sequence with $\lim\sup_{n}b_{n}\leq 0$. If 
\[
a_{n+1}\leq (1-\beta_{n})a_{n}+\beta_{n}b_{n}\qquad\forall n\in\N
\]
then $\lim_{n\to\infty}a_{n}=0$. 
\end{lemma}

\begin{lemma}\label{Saejung1} 
	\citep{saejung2012approximation}
	Let $\{a_n\}_{n\in\N}$ be a sequence of nonnegative real numbers, $\{\beta_n\}_{n\in\N}$ be a sequence of real numbers in $(0, 1)$ with $\sum_{n}\beta_n=\infty$ and $\{b_n\}_{n\in\N}$ be a sequence of real numbers. Assume that
	\begin{equation*}
	a_{n+1}\le (1-\beta_n)a_n+\beta_n b_n \text{ for all } n\geq 1.
	\end{equation*}
	If $\limsup_{k\to\infty} b_{n_k} \le 0$ for every subsequence $\{a_{n_k}\}_{k=1}^{\infty}$ of $\{a_n\}_{n=1}^{\infty}$ satisfying $\liminf_{k\to\infty}(a_{n_k+1}-a_{n_k})\geq 0$ then $\lim_{n\to \infty}{a_n} = 0$.
\end{lemma}

\begin{lemma}\label{lem:lambda}
Let $\tau_{0}>0$ and $\mu\in(0,1)$. Let $A:\scrH\to\scrH$ be a $L$-Lipschitz continuous operator. The sequence $\{\tau_{n}\}_{n\geq 0}$ generated by eq. \eqref{stepsize} is non-increasing and satisfies 
\begin{equation}\label{eq:LimLambda}
\lim_{n\to\infty}\tau_{n} = \tau \geq \bar{\tau} :=\min\left\{\frac{\mu}{L},\tau_{0}\right\}. 
\end{equation}
Furthermore, 
\begin{equation}
\norm{A(w_{n})-A(y_{n})}\leq \frac{\mu}{\tau_{n+1}}\norm{w_{n}-y_{n}}\qquad\forall n\geq 1. 
\end{equation}
\end{lemma}
\begin{proof}
Since 
$\tau_{n+1}=\min\left\{\frac{\mu\norm{w_{n}-y_{n}}}{\norm{A(w_{n})-A(y_{n})}},\tau_{n}\right\}$, 
it is clear that $\tau_{n+1}\leq\tau_{n}$ for all $n\geq 0$. Moreover, using the $L$-Lipschitz continuity of the operator $A$ gives
\[
\frac{\mu\norm{w_{n}-y_{n}}}{\norm{A(w_{n})-A(y_{n})}}\geq\frac{\mu}{L}\quad \text{ if }A(w_{n})\neq A(y_{n}). 
\]
Hence, $\tau_{n+1}\geq \min\{\frac{\mu}{L},\tau_{n}\}$ for all $n$. By induction, it follows that $\{\tau_{n}\}_{n}$ is bounded from below by  $\min\{\frac{\mu}{L},\tau_{0}\}.$ Therefore, $\lim_{n\to\infty}\tau_{n} = \tau \geq \bar{\tau} :=\min\left\{\frac{\mu}{L},\tau_{0}\right\}. $
\end{proof}

\subsection{Proof of Theorem \ref{th:main}}
\label{sec:ProofMain}
We start with an auxiliary technical result, which guarantees that weak cluster points of the algorithm produce solutions of \ac{DUE}. It is based on techniques from \cite{Vuong18}. 

\begin{lemma}\label{003}  Let $\{w_n\}$ be a sequence generated by Algorithm \ref{alg:IFBF}. If there exists a subsequence $\{w_{n_k}\}$  convergent weakly to $z\in\scrH$ and $\lim_{k\to\infty}\|w_{n_k}-y_{n_k}\|=0$, then, having Assumption \ref{ass:1}-\ref{ass:monotone} in place, we know $z\in\Omega.$
\end{lemma}
\begin{proof} 
Recall that $y_{n}=P_{\scrX}(w_{n}-\tau_{n}A(w_{n}))$. By Lemma \ref{lem:projection}(i), we have
\begin{equation*}
\langle w_{n_k}-\tau_{n_k} A(w_{n_k})-y_{n_k},x-y_{n_k}\rangle \le 0,\;\; \forall x\in \scrX,
\end{equation*}
or, equivalently,
\begin{equation*}
\dfrac{1}{\tau_{n_k}}\langle w_{n_k}-y_{n_k},x-y_{n_k}\rangle \le \langle A(w_{n_k}),x-y_{n_k}\rangle,\;\; \forall x\in \scrX.
\end{equation*}
Consequently, we have
\begin{equation}\label{aa3}
\dfrac{1}{\tau_{n_k}}\langle w_{n_k}-y_{n_k},x-y_{n_k}\rangle +\langle A(w_{n_k}),y_{n_k}-w_{n_k}\rangle  \le \langle A(w_{n_k}),x-w_{n_k}\rangle,\;\; \forall x\in \scrX.
\end{equation}
Since $\{w_{n_k}\}$ is weakly convergent, it is bounded. 
Then, by the Lipschitz continuity of $A$, $\{A(w_{n_k})\}$ is bounded. 
As $\|w_{n_k}-y_{n_k}\|\to 0$, $\{y_{n_k}\}$ is also bounded and $\tau_{n_k} \geq \min\{\tau_{0},\dfrac{\mu}{L}\}$. 
Passing (\ref{aa3}) to the limit as $k\to \infty$, we get
\begin{equation}\label{13}
\liminf_{k\to\infty}\langle A(w_{n_k}),x-w_{n_k}\rangle \geq 0,\;\;\forall x\in \scrX.
\end{equation}
Moreover, we have
\begin{align*}
\langle A(y_{n_k}),x-y_{n_k}\rangle&=\langle A(y_{n_k})- A(w_{n_k}),x-w_{n_k}\rangle+\langle A(w_{n_k}),x-w_{n_k}\rangle \\
&+\langle A(y_{n_k}),w_{n_k}-y_{n_k}\rangle\\
&\geq -\norm{A(y_{n_{k}})-A(w_{n_{k}})}\cdot\norm{x-w_{n_{k}}}+\inner{A(w_{n_{k}}),x-w_{n_{k}}}\\
&-\norm{A(y_{n_{k}})}\cdot\norm{w_{n_{k}}-y_{n_{k}}}.
\end{align*}
Since $\lim_{k\to\infty}\|w_{n_k}-y_{n_k}\|=0$ and $A$ is $L$-Lipschitz continuous on $H$, we get from the above 
\begin{equation*}
\lim_{k\to\infty}\|A(w_{n_k})-A(y_{n_k})\|=0.
\end{equation*} 
Together with (\ref{13}), we obtain
\begin{equation*}\label{aa5}
\liminf_{k\to\infty}\langle A(y_{n_k}),x-y_{n_k}\rangle \geq 0\qquad\forall x\in\scrX.
\end{equation*}
Next, we show that $z\in\Omega.$ We choose a sequence $\{\epsilon_k\}$ of positive numbers decreasing and tending to
$0$. For each $k\geq 1$, we denote by $N_k$ the smallest positive integer such that
\begin{equation}\label{aa7}
\langle A(y_{n_j}),x-y_{n_j}\rangle +\epsilon_k \geq 0,\;\; \forall j\geq N_k.
\end{equation}
Since $\{ \epsilon_k\}$ is decreasing, it is easy to see
that the sequence $\{N_k\}$ is increasing. Furthermore, for each $k\geq 1$, since $\{y_{N_k}\}\subset \scrX$, we can suppose $A(y_{N_k})\ne 0$ (otherwise, $y_{N_k}$ is a solution) and, setting
\begin{equation*}
v_{N_k} =
\dfrac{A(y_{N_k})}{\|A(y_{N_k})\|^2 },
\end{equation*}
we have $\langle A(y_{N_k}), v_{N_k}\rangle = 1$ for each $k\geq 1$. Now, we can deduce from (\ref{aa7}) that, for each $k\geq 1$,
\begin{equation*}
\langle A(y_{N_k}), x+\epsilon_k v_{N_k}-y_{N_k}\rangle=\inner{A(y_{N_{k}}),x-y_{n_{k}}}+\epsilon_{k} \geq 0.
\end{equation*}
Since $A$ is pseudo-monotone (Assumption \ref{ass:monotone}) on $\scrH$, we get
\begin{equation*}
\langle A(x+\epsilon_k v_{N_k}), x+\epsilon_k v_{N_k}-y_{N_k}\rangle \geq 0.
\end{equation*}
This implies that
\begin{equation}\label{aa8}
\langle A(x), x-y_{N_k}\rangle \geq \langle A(x)-A(x+\epsilon_k v_{N_k}), x+\epsilon_k v_{N_k}-y_{N_k} \rangle-\epsilon_k \langle A(x), v_{N_k}\rangle.
\end{equation}

Now, we show that $\lim_{k\to\infty}\epsilon_k v_{N_k}=0$. Indeed, since $w_{n_k}\rightharpoonup z$ and $\lim_{k\to\infty}\|w_{n_k}-y_{n_k}\|=0,$ we obtain  $y_{N_k}\rightharpoonup z  \text{ as } k \to \infty$. Since $\{y_n\}\subset \scrX$, we clearly have $z\in \scrX$ as well. Since
$A$ is sequentially weakly continuous on $\scrX$, $\{ A(y_{n_k})\}$ converges weakly to $A(z)$. We can suppose $A(z) \ne 0$ (otherwise, $z$ is a solution). Since the norm mapping is
sequentially weakly lower semi-continuous, we have
\begin{equation*}
0 < \|A(z)\|\le \liminf_{k\to \infty}\|A(y_{n_k})\|.
\end{equation*}
Together with $\{y_{N_k}\}\subset \{y_{n_k}\}$ and $\epsilon_k \to 0$ as $k \to \infty$, we readily conclude
\begin{align*}
0 \le \limsup_{k\to\infty} \|\epsilon_k  v_{N_k} \|= \limsup_{k\to\infty} \left(\dfrac{\epsilon_k}{\|A(y_{n_k})\|}\right)\le \dfrac{\limsup_{k\to\infty}\epsilon_k } {\liminf_{k\to\infty}\|A(y_{n_k})\|}=0.
\end{align*}
Hence, $\lim_{k\to\infty} \epsilon_k v_{N_k} = 0$.

Now, letting $k\to \infty$, then the right hand side of (\ref{aa8}) tends to zero by $A$ is uniformly continuous, $\{w_{N_k}\}, \{v_{N_k}\}$ are bounded and  $\lim_{k\to\infty}\epsilon_k v_{N_k}=0$. Thus, we get
\begin{equation*}
\liminf_{k\to\infty}\langle A(x),x-y_{N_k}\rangle \geq 0.
\end{equation*}
Hence, for all $x\in \scrX$, we have
\begin{equation*}
\langle A(x), x-z\rangle=\lim_{k\to\infty} \langle A(x), x-y_{N_k}\rangle =\liminf_{k\to\infty} \langle A(x), x-y_{N_k}\rangle \geq 0.
\end{equation*}
By Lemma \ref{002}, $z\in\Omega$. This completes the proof.
\end{proof}

The next result established the boundedness of the sequence of path flows. 
%%%%%%%%%%%%%%%
\begin{lemma}
\label{Claim1}
Let Assumptions \ref{ass:1}-\ref{ass:monotone} hold. The sequence of path flows $\{h_n\}_{n=1}^\infty$ generated by Algorithm \ref{alg:IFBF} is bounded. In addition,
\begin{equation}\label{xx11}
\|h_{n+1}-h^*\|^2\le \|w_n-h^*\|^2-\lambda\Big(1-\mu\dfrac{\tau_n}{\tau_{n+1}}\Big)\Big(2-\lambda+\lambda\mu \dfrac{\tau_n}{\tau_{n+1}} \Big)\|y_n-w_n\|^2.
\end{equation}
\end{lemma}
\begin{proof}
We have
\begin{align}
\|h_{n+1}-h^*\|^2&= \|(1-\lambda)w_n +\lambda (y_n-\tau_n(A(y_n)-A(w_n)))-h^*\|^2\notag\\
&= \|(1-\lambda)(w_n -h^*)+\lambda (y_n-h^*)+\lambda\tau_n(A(w_n)-A(y_n)))\|^2\notag\\
&=(1-\lambda)^2\|w_n-h^*\|^2+\lambda^2\|y_n-h^*\|^2+\lambda^2\tau_n^2\|A(w_n)-A(y_n)\|^2\notag\\
&+2(1-\lambda)\lambda\langle w_n-h^*,y_n-h^*\rangle +2(1-\lambda)\lambda\tau_n\langle w_n-h^*,A(w_n)-A(y_n)\rangle \notag \\
&+2\lambda^2\tau_n\langle y_n-h^*,A(w_n)-A(y_n)\rangle\label{t3324}.
\end{align}
Combining
\begin{equation}\label{aabc}
2\langle w_n-h^*,y_n-h^*\rangle=\|w_n-h^*\|^2+\|y_n-h^*\|^2-\|w_n-y_n\|^2,
\end{equation}
with the definition of $\{\tau_n\}$, it is easy to see that 
\begin{equation}\label{xxxv1}
\|A(w_n)-A(y_n)\|\le \dfrac{\mu}{\tau_{n+1}}\|w_n-y_n\|,\,\,\,  \forall n\geq 0.
\end{equation}
Substituting \eqref{aabc} and \eqref{xxxv1} into \eqref{t3324}, we get
\begin{align}
\|h_{n+1}-h^*\|^2&\le (1-\lambda)\|w_n-h^*\|^2+\lambda\|y_n-h^*\|^2+\lambda^2\dfrac{\tau_n^2}{\tau^2_{n+1}}\mu^2\|w_n-y_n\|^2-(1-\lambda)\lambda\| w_n-y_n\|^2 \notag\\
&  +2(1-\lambda)\lambda\tau_n\langle w_n-h^*,A(w_n)-A(y_n)\rangle+2\lambda^2\tau_n\langle y_n-h^*,A(w_n)-A(y_n)\rangle\label{abc45}.
\end{align}
Lemma \ref{lem:projection}(i) yields the estimate
\begin{align*}
\|y_n-h^*\|^2&=\inner{y_{n}-h^{\ast},y_{n}-h^{\ast}}\\
&=\inner{P_{\scrX}(w_{n}-\tau_{n}A(w_{n}))-P_{\scrX}(h^{\ast}),P_{\scrX}(w_{n}-\tau_{n}A(w_{n}))-P_{\scrX}(h^{\ast})}\\
&=\inner{y_{n}-h^{\ast},w_{n}-\tau_{n}A(w_{n})-h^{\ast}}+\inner{P_{\scrX}(w_{n}-\tau_{n}A(w_{n})-P_{\scrX}(h^{\ast}),P_{\scrX}(w_{n}-\tau_{n}A(w_{n}))-w_{n}+\tau_{n}A(w_{n})}\\
&\le \langle y_n-h^*,w_n-\tau_n A(w_n)-h^{\ast}\rangle\\
&=\dfrac{1}{2}\|y_n-h^*\|^2+\dfrac{1}{2}\|w_n-\tau_n A(w_n)-h^*\|^2-\dfrac{1}{2}\|(w_n-h^*)-(w_n-\tau_n A(w_n))\|^2\\
&=\dfrac{1}{2}\|y_n-h^*\|^2+\dfrac{1}{2}\|w_n-h^*\|^2-\dfrac{1}{2}\|y_n-w_n\|^2-\langle y_n-h^*,\tau_n A(w_n)\rangle, 
\end{align*}
or equivalently
\begin{align}\label{abc1}
\|y_n-h^*\|^2 \le \|w_n-h^*\|^2-\|y_n-w_n\|^2-2\langle y_n-h^*,\tau_n A(w_n)\rangle.
\end{align}
Since $h^*\in \Omega$, we have $\langle A(h^*),y_n - h^*\rangle \geq 0$. It follows from the pseudo-monotonicity of $A$ that
\begin{equation}\label{abc2}
2\langle \tau_n A(y_n),y_n - h^*\rangle \geq 0.
\end{equation}
Adding \eqref{abc1} and \eqref{abc2}, we obtain
\begin{align}\label{abc3}
\|y_n-h^*\|^2 \le \|w_n-h^*\|^2-\|y_n-w_n\|^2-2\tau_n\langle y_n-h^*, A(w_n)-A(y_n)\rangle.
\end{align}
Substituting \eqref{abc3} into \eqref{abc45}, we get
\begin{align}
\|h_{n+1}-h^*\|^2
&\le 
(1-\lambda)\|w_n-h^*\|^2+\lambda\|w_n-h^*\|^2-\lambda\|y_n-w_n\|^2-2\lambda\tau_n\langle y_n-h^*, A(w_n)-A(y_n)\rangle\notag\\
&\quad +\lambda^2\dfrac{\tau_n^2}{\tau^2_{n+1}}\mu^2\|w_n-y_n\|^2-(1-\lambda)\lambda\| w_n-y_n\|^2 \notag\\
&  \quad +2(1-\lambda)\lambda\tau_n\langle w_n-h^*,A(w_n)-A(y_n)\rangle+2\lambda^2\tau_n\langle y_n-h^*,A(w_n)-A(y_n)\rangle\notag\\
&= 
(1-\lambda)\|w_n-h^*\|^2+\lambda\|w_n-h^*\|^2-\lambda\|y_n-w_n\|^2-2\lambda\tau_n\langle y_n-h^*, A(w_n)-A(y_n)\rangle\notag\\
&\quad +\lambda^2\dfrac{\tau_n^2}{\tau^2_{n+1}}\mu^2\|w_n-y_n\|^2-(1-\lambda)\lambda\tau_n\| w_n-y_n\|^2 \notag\\
&\quad   +2(1-\lambda)\lambda\tau_n\langle y_n-h^*,A(w_n)-A(y_n)\rangle\notag\\
&\quad +2(1-\lambda)\lambda\tau_n\langle w_n-y_n,A(w_n)-A(y_n)\rangle+2\lambda^2\tau_n\langle y_n-h^*,A(w_n)-A(y_n)\rangle\notag\\
&= \|w_n-h^*\|^2\notag\\
&\quad +\lambda^2\dfrac{\tau_n^2}{\tau^2_{n+1}}\mu^2\|w_n-y_n\|^2-(2-\lambda)\lambda\| w_n-y_n\|^2 \notag\\
&\quad +2(1-\lambda)\lambda\tau_n\langle w_n-y_n,A(w_n)-A(y_n)\rangle\notag\\
%&\le \|w_n-h^*\|^2\notag\\
%&\quad +\lambda^2\dfrac{\tau_n^2}{\tau^2_{n+1}}\mu^2\|w_n-y_n\|^2-(2-\lambda)%\lambda\| w_n-y_n\|^2 \notag\\
%&\quad +2(1-\lambda)\lambda\tau_n\langle w_n-y_n,A(w_n)-A(y_n)\rangle \notag\\
&\le \|w_n-h^*\|^2\notag\\
&\quad+\lambda^2\dfrac{\tau_n^2}{\tau^2_{n+1}}\mu^2\|w_n-y_n\|^2-(2-\lambda)\lambda\| w_n-y_n\|^2 \notag\\
&\quad+2(1-\lambda)\lambda\dfrac{\tau_n}{\tau_{n+1}}\mu \|w_n-y_n\|^2\notag\\
&= \|w_n-h^*\|^2-\lambda\bigg[2-\lambda -\lambda \dfrac{\tau^2_n}{\tau^2_{n+1}}\mu^2-2(1-\lambda)\dfrac{\tau_n}{\tau_{n+1}}\mu\bigg]\|w_n-y_n\|^2\notag\\
 &=\|w_n-h^*\|^2-\lambda\Big(1-\mu\dfrac{\tau_n}{\tau_{n+1}}\Big)\Big(2-\lambda+\lambda\mu \dfrac{\tau_n}{\tau_{n+1}} \Big)\|y_n-w_n\|^2
\label{abc4}.
\end{align}
Since
$$
\lim_{n\to\infty}\Big(1-\mu\dfrac{\tau_n}{\tau_{n+1}}\Big)\Big(2-\lambda+\lambda\mu \dfrac{\tau_n}{\tau_{n+1}} \Big)=(1-\mu)(2-\lambda+\lambda \mu)>0
$$
there exists $n_0\in \mathbb{N}$ such that 
$$\Big(1-\mu\dfrac{\tau_n}{\tau_{n+1}}\Big)\Big(2-\lambda+\lambda\mu \dfrac{\tau_n}{\tau_{n+1}} \Big)>0 \ \ \forall n\geq n_0.
$$
Hence
\begin{equation}\label{v1v}
\|h_{n+1}-h^*\|\le \|w_n-h^*\|  \ \ \forall n\geq n_0.
\end{equation}
On the one hand, using the definition of $w_n$, we  obtain
\begin{align}\label{vvvx2}
\|w_n-h^*\|&=\|(1-\beta_n)(h_n+\alpha_n(h_n-h_{n-1}))-h^*\|\notag\\
&=\|(1-\beta_n)(h_n-h^*)+(1-\beta_n)\alpha_n(h_n-h_{n-1})-\beta_n h^*\|\notag\\
&\le (1-\beta_n)\|h_n-h^*\|+(1-\beta_n)\alpha_n\|h_n-h_{n-1}\|+\beta_n \|h^*\|\notag\\
&=(1-\beta_n)\|h_n-h^*\|+\beta_n [ (1-\beta_n)\dfrac{\alpha_n}{\beta_n}\|h_n-h_{n-1}\|+\|h^*\|].
\end{align}
From \eqref{vvxvx1} and \eqref{para}, we have
\begin{equation*}
\dfrac{\alpha_n}{\beta_n}\|h_n-h_{n-1}\|\le \dfrac{\epsilon_n}{\beta_n}\to 0.
\end{equation*}
Hence,
$$
\lim_{n\to\infty}\bigg[ (1-\beta_n)\dfrac{\alpha_n}{\beta_n}\|h_n-h_{n-1}\|+\|h^*\|\bigg]=\|h^*\|,
$$
Therefore, there exists $M>0$ such that
\begin{equation}\label{vvvx1}
(1-\beta_n)\dfrac{\alpha_n}{\beta_n}\|h_n-h_{n-1}\|+\|h^*\|\le M.
\end{equation}
Combining \eqref{vvvx2} and \eqref{vvvx1} we obtain
\begin{equation}\label{v2v}
\|w_n-h^*\|\le (1-\beta_n)\|h_n-h^*\|+\beta_n  M.
\end{equation}
Hence, from \eqref{v1v} and \eqref{v2v}, we have
\begin{align*}
\|h_{n+1}-h^*\|\le &(1-\beta_n)\|h_n-h^*\|+\beta_n  M\\
&=\max\{\|h_n-h^*\|, M \}\leq \ldots\le \max\{\|h_{n_0}-h^*\|,M\}.
\end{align*}
Therefore, the sequence $\{h_n\}_{n=1}^\infty$ is bounded.
\end{proof}
%%%%%%%%%%%%%%%%%%%
\begin{lemma}\label{Claim2}
It holds that
\begin{align*}
\lambda\Big(1-\mu\dfrac{\tau_n}{\tau_{n+1}}\Big)\Big(2-\lambda+\lambda\mu \dfrac{\tau_n}{\tau_{n+1}} \Big)\|w_n-y_n\|^2\le\|h_n-h^*\|^2-\|h_{n+1}-h^*\|^2+\beta_n M_1.
\end{align*}
\end{lemma}
\begin{proof}
Eq. \eqref{v2v} yields 
\begin{align}\label{vvvx3}
\|w_n-h^*\|^2&\le (1-\beta_n)^2\|h_n-h^*\|^2+2\beta_n(1-\beta_n)M \|h_n-h^*\|+\beta_n^2M^2\notag\\
&\le \|h_n-h^*\|^2+\beta_n [2(1-\beta_n)M \|h_n-h^*\|+\beta_n M^2]\notag\\
&\le \|h_n-h^*\|^2+\beta_n M_1,
\end{align}
where $M_1:=\max \{2(1-\beta_n)M \|h_n-h^*\|+\beta_n M^2:\  n\in \mathbb{N}\}$. Substituting \eqref{vvvx3} into \eqref{abc4} we get
\begin{align*}\label{vvvx4}
\|h_{n+1}-h^*\|^2&\le \|h_n-h^*\|^2+\beta_n M_1-\lambda\Big(1-\mu\dfrac{\tau_n}{\tau_{n+1}}\Big)\Big(2-\lambda+\lambda\mu \dfrac{\tau_n}{\tau_{n+1}} \Big)\|w_n-y_n\|^2,
\end{align*}
or equivalently
\begin{align*}
\lambda\Big(1-\mu\dfrac{\tau_n}{\tau_{n+1}}\Big)\Big(2-\lambda+\lambda\mu \dfrac{\tau_n}{\tau_{n+1}} \Big)\|w_n-y_n\|^2
\le\|h_n-h^*\|^2-\|h_{n+1}-h^*\|^2+\beta_n M_1.
\end{align*}

\end{proof}
%%%%%%%%%%%%%%%%%%%%%%
\begin{lemma}\label{Claim3}
It holds that 
\begin{align*}
\|h_{n+1}-h^*\|^2\le &(1-\beta_n)\|h_n-h^*\|^2+\beta_n\bigg[2(1-\beta_n)\|h_n-h^*\|\dfrac{\alpha_n}{\beta_n} \|h_n-h_{n-1}\|\\
&+\dfrac{\alpha^{2}_n}{\beta_n}\|h_n-h_{n-1}\|^{2}+2 \|h^*\|\cdot \|w_n-h_{n+1}\|  +2 \langle -h^*,h_{n+1}-h^* \rangle  \bigg].
\end{align*}
\end{lemma}
\begin{proof}
Using the inequalities \eqref{v1v} and then \eqref{xx22} as well as $\beta_{n}\in(0,1)$, we get
\begin{align*}
\|h_{n+1}-h^*\|^2\le& \|w_n-h^*\|^2\\
=&\|(1-\beta_n) (h_n-h^*)+(1-\beta_n) \alpha_n (h_n-h_{n-1})-\beta_n h^*\|^2\\
\le& \|(1-\beta_n)(h_n-h^*)+(1-\beta_n)\alpha_n (h_n-h_{n-1})\|^2+2\beta_n \langle -h^*,w_n-h^*\rangle \\
&=(1-\beta_{n}^{2}\norm{h_{n}-h^{\ast}}^{2}+\alpha^{2}_{n}(1-\beta_{n})^{2}\norm{h_{n}-h_{n-1}}^{2}+2\alpha_{n}(1-\beta_{n})^{2}\inner{h_{n}-h^{\ast},h_{n}-h_{n-1}}\\
&+2\beta_{n}\inner{-h^{\ast},w_{n}-h^{\ast}}\\
\leq& (1-\beta_n)\|h_n-h^*\|^2+2(1-\beta_n)\alpha_n\|h_n-h^*\| \|h_n-h_{n-1}\|+\alpha_n^2\|h_n-h_{n-1}\|^2\\
&+2\beta_{n} \langle -h^*,w_n-h_{n+1} \rangle+2\beta_{n} \langle -h^*,h_{n+1}-h^* \rangle  \\
\le &(1-\beta_n)\|h_n-h^*\|^2+\beta_n\bigg[2(1-\beta_n)\|h_n-h^*\|\dfrac{\alpha_n}{\beta_n} \|h_n-h_{n-1}\|\\
&+\dfrac{\alpha^{2}_n}{\beta_n}\|h_n-h_{n-1}\|^{2}+2 \| h^*\| \cdot \|w_n-h_{n+1}\| +2 \langle -h^*,h_{n+1}-h^* \rangle  \bigg].
\end{align*}

\end{proof}
%%%%%%%%%%%%%%%%%%%%%%%%%%
Equipped with these preliminary result, we are now ready to proof the main result of this paper, Theorem \ref{th:main}. We restate the theorem below again, for the readers' convenience.  

\begin{theorem}\label{theorem1}
Let Assumptions \ref{ass:1}-\ref{ass:monotone} hold. The sequence of path flows $\{h_n\}_{n=1}^\infty$ generated by Algorithm \ref{alg:IFBF} converges strongly to an element $h^{\ast}\in\Omega$, where $h^{\ast}=\argmin\{\norm{z}:z\in\Omega\}.$
\end{theorem}
\begin{proof} 
We split the proof in two cases. 
Let us call $a_{n}:=\norm{h_{n}-h^{\ast}}^{2}$ and 
\begin{align*}
b_{n}&:=2(1-\beta_n)\|h_n-h^*\|\dfrac{\alpha_n}{\beta_n} \|h_n-h_{n-1}\|\\
&+ \dfrac{\alpha^{2}_n}{\beta_n}\|h_n-h_{n-1}\|^{2}+2 \| h^*\| \cdot \|w_n-h_{n+1}\| +2 \langle -h^*,h_{n+1}-h^* \rangle ,
\end{align*}
so that Lemma \ref{Claim3} boils down to the recursion
\[
a_{n+1}\leq(1-\beta_{n})a_{n}+\beta_{n}b_{n}.
\]
We split the proof in two case.\\
{\bf Case 1:} There exists $N_{0}\in\N$ such that $a_{n+1}\leq a_{n}$ for all $n\geq N_{0}$. Then, $\lim\inf_{n\to\infty}(a_{n+1}-a_{n})=0$, and it follows $\lim\sup_{n}b_{n}\leq 0$. The conclusion follows from Lemma \ref{lem:summation}. \\
{\bf Case 2:} By Lemma \ref{Saejung1}, it suffices to show that 
\[
\limsup_{k\to\infty}\langle -h^*, h_{n_k+1}-h^*\rangle\le 0
\]
for every subsequence $\{\|h_{n_k}-h^*\|\}_{k=1}^\infty$ of $\{\|h_n-h^*\|\}_{n=1}^\infty$ satisfying
\begin{equation*}
\liminf_{k\to\infty}(\|h_{n_k+1}-h^*\|-\|h_{n_k}-h^*\|)\geq 0.
\end{equation*}

For this, suppose that $\{\|h_{n_k}-h^*\|\}_{k=1}^\infty$ is a subsequence of $\{\|h_n-h^*\|\}_{n=1}^\infty$ such that
$\liminf_{k\to\infty}(\|h_{n_k+1}-h^*\|-\|h_{n_k}-h^*\|)\geq 0.$ Then
\begin{equation*}
\liminf_{k\to\infty}(\|h_{n_k+1}-h^*\|^2-\|h_{n_k}-h^*\|^2)=\liminf_{k\to\infty}[(\|h_{n_k+1}-h^*\|-\|h_{n_k}-h^*\|)(\|h_{n_k+1}-h^*\|+\|h_{n_k}-h^*\|)]\geq 0.
\end{equation*}
By Lemma \ref{Claim2} we obtain
\begin{align*}
\limsup_{k\to\infty}&\bigg[\lambda\Big(1-\mu\dfrac{\tau_{n_k}}{\tau_{n_k+1}}\Big)\Big(2-\lambda+\lambda\mu \dfrac{\tau_{n_k}}{\tau_{n_k+1}} \Big)\|w_{n_k}-y_{n_k}\|^2\bigg]\\
&\le\limsup_{k\to\infty}\bigg[ \|h_{n_k}-h^*\|^2-\|h_{{n_k}+1}-h^*\|^2+\beta_{n_k} M_1\bigg]\\
&\le \limsup_{k\to\infty}\bigg[ \|h_{n_k}-h^*\|^2-\|h_{n_k+1}-h^*\|^2\bigg]+\limsup_{k\to\infty}\beta_{n_k} M_1\\
&=-  \liminf_{k\to\infty}\bigg[ \|h_{n_k+1}-h^*\|^2-\|h_{_{n_k}}-h^*\|^2\bigg]
\le 0.
\end{align*}
This implies that \begin{equation}\label{h1v}
\lim_{k\to\infty} \|y_{n_k}-w_{n_k}\|=0.
\end{equation}
On the other hand, we have
\begin{align}
\|h_{n+1}-y_n\|&=\|(1-\lambda)(w_n-y_n)+\lambda \tau_n (Ay_n-Aw_n)\|\notag\\ 
&\le (1-\lambda)\|w_n-y_n\|+\lambda\tau_n\|Ay_n-Aw_n\|\notag\\
&\le (1-\lambda)\|w_n-y_n\|+\lambda \dfrac{\tau_n}{\tau_{n+1}}\|w_n-y_n\|\notag\\
&=(1-\lambda+\lambda\mu\dfrac{\tau_n}{\tau_{n+1}})\|w_n-y_n\|. \label{h2v}
\end{align}
Combining \eqref{h1v} and \eqref{h2v} we get
\begin{equation}\label{h3v}
\lim_{k\to\infty} \|h_{n_k+1}-y_{n_k}\|=0.
\end{equation}
Also from \eqref{h1v} and \eqref{h3v}, it holds
\begin{equation}\label{ttt5}
\lim_{k\to\infty} \|h_{n_k+1}-w_{n_k}\|=0.
\end{equation}
Now, we show that
\begin{equation}\label{ttt8}
\|h_{n_k+1}-h_{n_k}\| \to 0\,\, \text{  as }\,\, n\to\infty.
\end{equation}
Indeed, we have
\begin{equation}\label{tttti1}
\|h_{n_k}-w_{n_k}\|=\alpha_{n_k} \|h_{n_k}-h_{n_k-1}\|=\beta_{n_k}\cdot\dfrac{\alpha_{n_k}}{\beta_{n_k}} \|h_{n_k}-h_{n_k-1}\|\to 0.
\end{equation}
From (\ref{ttt5}) and (\ref{tttti1}), we get
\begin{equation*}
\|h_{n_k+1}-h_{n_k}\|\le \|h_{n_k+1}-w_{n_k}\|+\|w_{n_k}-h_{n_k}\|\to 0.
\end{equation*}
 Since the sequence $\{h_{n_k}\}_{k=1}^\infty$ is bounded, it follows that there exists a subsequence
   $\{h_{n_{k_j}}\}_{j=1}^\infty$ of $\{h_{n_k}\}_{k=1}^\infty$, which converges weakly to some $z^*\in H$, such that
\begin{equation}\label{444t}
\limsup_{k\to \infty}\langle -h^*,h_{n_k}-h^*\rangle =\lim_{j\to \infty}\langle -h^*,h_{n_{k_j}}-h^*\rangle=\langle -h^*,z^*-h^*\rangle.
\end{equation}
From (\ref{tttti1}), we obtain
\begin{equation*}
w_{n_k} \rightharpoonup z^* \text{ as } k\to \infty.
\end{equation*}
Using Lemma \ref{003}, we conclude 
$
z^*\in \Omega. % = {\rm Sol(\scrX,F)}.
$
Next, since (\ref{444t}) and the definition of $h^*=P_{\Omega}(0)$, we have
\begin{equation}\label{555t}
\limsup_{k\to \infty}\langle -h^*,h_{n_k}-h^*\rangle =\langle -h^*,z^*-h^*\rangle\le 0.
\end{equation}
Combining  (\ref{ttt8}) and (\ref{555t}), we have
\begin{align}\label{5555t}
\limsup_{k\to \infty}\langle -h^*,h_{n_k+1}-h^*\rangle &
\le \limsup_{k\to \infty}\langle -h^*,h_{n_k}-h^*\rangle\notag\\
&=\langle -h^*,z^*-h^*\rangle\le 0.
\end{align}
Hence, by (\ref{5555t}), $\lim_{n\to\infty}\dfrac{\alpha_n}{\beta_n}\|h_n-h_{n-1}\|=0$, $\lim_{n\to\infty}\|h_{n+1}-w_n\|=0$, Lemma \ref{Claim3} and Lemma \ref{Saejung1}, we obtain the desired result, $\lim_{n\to\infty}\|h_n-h^*\|=0$.
\end{proof}

%----------------------------------------------------------------------
%%% CONCLUSIONS
%----------------------------------------------------------------------
\section{Conclusion}
\label{sec:conclusion}
This paper builds on recent advances in the computational theory of dynamic user equilibrium. Building on the network extension of the LWR model and its formulation in terms of a system of differential algebraic equations. Our aim is to advocate the use of strongly convergent fixed point iterations for computing dynamic user equilibrium which are provably convergent under mild a-priori monotonicity assumptions on the path delay operator, and which are adaptive in the sense that no global bound on the Lipschitz constant needs to be known. We focussed on the construction of new strongly convergent forward-backward-forward algorithms, augmented by relaxation and inertial modifications. We tested the performance of our algorithms in the Nguyen and the Sioux falls network, and provide thereby evidence that our methods improve upon pre-implemented solvers. In future research we aim to improve the fixed point iteration by reducing its complexity in terms of calls of the DNL subroutine. Indeed, the price to pay for provable convergence under weaker assumptions is that under FBF splitting we have to evaluate the delay operator twice per iterations, which is computationally costly. The FB iteration needs only a single call of the DNL, but converges only under strong monotonicity assumptions. In a future publication we will describe a single-call variant of the FBF, which still guarantees strong convergence under the same monotonicity assumptions as used in this paper. 

Another very interesting direction of research we plan to pursue is to replace the costly DNL procedure by alternative approximation schemes motivated by machine learning approaches as commenced in \cite{Song2017}. 

%\begin{acknowledgements}
%We thank Professor Terry L. Friesz for insightful discussion. Constructive feedback from two anonymous referees is also gratefully acknowledged.
%\end{acknowledgements}



%\bibliographystyle{spbasic} 
%\bibliography{mybib}

%\end{thebibliography}

\end{document}